\RequirePackage{fix-cm}

\documentclass{svjour3}

\expandafter\let\csname proof\endcsname=\relax
\expandafter\let\csname endproof\endcsname=\relax

\expandafter\let\csname remark\endcsname=\relax
\expandafter\let\csname endremark\endcsname=\relax

\makeatletter
	\expandafter\let\csname c@remark\endcsname=\relax
\makeatother

\usepackage{amsmath,amssymb,amsthm}
\usepackage{bm,color,mathrsfs}

\usepackage{cite}

\usepackage{hyperref}

\newcommand{\bbR}{\mathbb R}
\newcommand{\CL}{\mathcal L}

\newcommand{\scC}{\mathscr C}
\newcommand{\scF}{\mathscr F}
\newcommand{\scM}{\mathscr M}
\newcommand{\scT}{\mathscr T}

\newcommand{\beq}{\begin{equation}}
\newcommand{\eeq}{\end{equation}}
\newcommand{\beqNo}{\begin{equation*}}
\newcommand{\eeqNo}{\end{equation*}}

\newcommand{\albe}{{\alpha \beta}}
\newcommand{\mnu}{{\mu \nu}}

\newcommand{\form}{\omega}
\newcommand{\orthproj}{{\bm \gamma^\sharp}}
\newcommand{\shl}{M}

\newcommand{\connect}{{\bm \nabla}}
\newcommand{\torsion}{{\bf T}}
\newcommand{\tangm}{\scT_p(\scM)}
\newcommand{\cotangm}{\scT_p^\star(\scM)}

\newcommand{\fframes}{\scF_{\bm u}}
\newcommand{\connectf}{{\bf D}}
\newcommand{\torsionf}{{\bf T}[\connectf]}
\newcommand{\tangf}{\scT_p(\fframes)}
\newcommand{\cotangf}{\scT_p^\star(\fframes)}

\newcommand{\basis}{{\bm e}}
\newcommand{\nbasis}{\bm{\mathrm e}}
\newcommand{\sbasis}[1]{{\basis^\star_{#1}}}
\newcommand{\scobasis}[1]{{\basis_\star^{#1}}}

\newcommand{\kroneck}[2]{\delta^{#1}_{\phantom{#1} #2}}
\newcommand{\proj}[2]{\gamma^{#1}_{\phantom{#1} #2}}

\newcommand{\coef}[3]{\Gamma^{#1}_{\phantom{#1} #2 #3}}
\newcommand{\struct}[3]{C^{#1}_{\phantom{#1} #2 #3}}

\newcommand{\coeff}[3]{\Gamma[D]^{#1}_{\phantom{#1} #2 #3}}
\newcommand{\structf}[3]{C[D]^{#1}_{\phantom{#1} #2 #3}}
\newcommand{\torscoeff}[3]{\mathrm{T}[D]^{#1}_{\phantom{#1} #2 #3}}

\theoremstyle{plain}
	\newtheorem{propositionS}{Proposition}
	\newtheorem{propositionE}{Proposition}
	\newtheorem{corollaryS}{Corollary}

\theoremstyle{remark}
	\newtheorem{remark}{Remark}
	\newtheorem{remarkapp}{Remark}

\makeatletter
	\@addtoreset{remark}{section}
	\@addtoreset{remarkapp}{section}
\makeatother

\definecolor{myblue}{rgb}{0.1,0.1,0.9}
\definecolor{mygreen}{rgb}{0.1,0.55,0.1}

\hypersetup{
	pdftitle={On the 1+3 Formalism in General Relativity},
	pdfauthor={Xavier Roy},
	pdfcreator={Xavier Roy},
	colorlinks=true,
	linkcolor=myblue,
	citecolor=mygreen,
}

\font\bigastfont=cmr10 scaled \magstep 1
\newcommand{\bdot}{\hbox{\bigastfont .}}

\begin{document}

\title{On the \texorpdfstring{$1+3$}{1+3} Formalism in General Relativity}

\author{Xavier Roy}
\institute{Astrophysics, Cosmology and Gravity Centre, \\
	Department of Mathematics and Applied Mathematics, \\
	Cape Town University, Rondebosch, 7701, South Africa \\
	\email{xavier.roy@uct.ac.za}}

\date{Received: date / Accepted: date}

\maketitle

\begin{abstract}
	We present in this paper the formalism for the splitting of a four-dimensional Lorentzian manifold by a set of time-like integral 
	curves. Introducing the geometrical tensors characterizing the local spatial frames induced by the congruence (namely, the spatial 
	metric tensor, the extrinsic curvature tensor and the Riemann curvature tensor), we derive the Gauss, Codazzi and Ricci equations, 
	along with the evolution equation for the spatial metric. In the present framework, the spatial frames do not form any hypersurfaces 
	as we allow the congruence to exhibit vorticity. The splitting procedure is then applied to the Einstein field equation and it results 
	in an equivalent set of constraint and evolution equations. We discuss the resulting systems and compare them with the ones obtained 
	from the 3+1 formalism, where the manifold is foliated by means of a family of three-dimensional space-like surfaces.
	
	
\end{abstract}


\section{Introduction}

\subsection{\texorpdfstring{$3+1$ and $1+3$}{3+1 and 1+3} splittings}

The Einstein field equation, 
\beqNo
	{}^{4\!}\bm R - \frac{1}{2} {}^{4\!}R \, \bm g + \Lambda \, \bm g 
		= \frac{8 \pi G}{c^4} \bm T \, , 
\eeqNo
expresses from a four-dimensional point of view the dynamical coupling between the geometry and the content of the Universe. 
Recasting this relation into a set of constraint and evolution equations in a three-dimensional framework allows for a more familiar 
and intuitive examination of the physical system at stake. It can bring for instance better understanding about the kinematical and 
geometrical aspects of particular relativistic solutions for which the field equations can be directly derived. It can also be of use, for 
instance, in the post-Newtonian treatment of weak gravitational fields and in the study of nonlinear structure formation in cosmology. 

This reformulation is realized upon splitting the space--time into space and time, and then upon applying the prescription to Einstein's 
general relativity. The splitting procedure itself can be carried out under two different perspectives: either by using the so-called 3+1 
formalism, or \textit{slicing} of space--time \cite{darmois27,lich39,lich44,lich52,foures48,foures52,foures56,smarr78,alcub08,baumg10}, 
or by using the so-called 1+3 formalism, or \textit{threading} of space--time \cite{zelm56,zelm59,catt58,catt59a,catt59b,catt59c,catt61,%
catt-gasp61,catt-gasp63,ferrarese63,ferrarese65,massa74a,massa74b,massa74c,jantzen91,jantzen92,boersma94a,boersma94b,boersma94c,%
bini98,bini12}.\footnote{%
The reader is referred to the comprehensive reviews \cite{gourg:book} and \cite{jantzen:book} for an introduction (and a historical 
presentation) of the 3+1 and 1+3 approaches, respectively.
}

The first of these techniques is based upon the introduction of a preferred family of three-dimensional space-like surfaces, through the 
level sets of some scalar field. It provides a global space-like association of points. No other relation is assumed, and the identification 
of points located on different slices can be performed arbitrarily.\footnote{%
This approach is also sometimes entitled `ADM formalism' (for Arnowitt, Deser and Misner \cite{adm62}). This denomination should 
be reserved, however, for the Hamiltonian formulation of general relativity only (see \cite{gourg:book} for further comments).
}

The second class of splitting is built upon a set of time-like integral curves, and it affords a global time-like relation between points. 
No other condition is supposed; however, in the case where the congruence exhibits vorticity, the local spatial frames orthogonal to the 
integral curves do not form any family of hypersurfaces. 

%

When both the space-like and time-like conditions hold, the two splittings can be applied on the manifold. In the generic configuration, 
the manifold is covered by a family of space-like hypersurfaces and an independent set of time-like curves. The simplest case is made of 
a vorticity-free time-like congruence, the orthogonal frames of which globally form spatial hypersurfaces.

In the framework of general relativity, an alternative to the above procedures consists in splitting directly the Einstein field equation 
with the aid of the kinematics of the fluid filling the space--time. This prescription disregards the geometry of the spatial frames and 
places emphasis instead on the kinematical quantities of the fluid congruence (namely, on its expansion rate, shear and vorticity) 
\cite{synge37,ehlers61,kundt62,ellis71,ellis73,maccallum73,elst97,ellis98}. The drawback of this approach lies in the impossibility of 
defining geometrically the extrinsic and Riemann curvatures of the spatial frames, and this change of perspective (from the kinematical 
to the geometrical point of view) can be requested as it constitutes an important aspect of Einstein's theory. In the threading picture, 
on the contrary, this equivalence (from the geometrical to the kinematical point of view) is straightforward.

The present article offers a detailed description of the threading of a four-dimensional manifold. We shall provide the Gauss, Codazzi 
and Ricci equations associated with the congruence of curves, along with the evolution equation for the metric of the local spatial frames. 
We shall then apply the splitting to the Einstein field equation. As we shall see, the resulting sets of 1+3 equations closely resemble their 
3+1 counterparts, with additional terms stemming from the non hypersurface-forming character of the spatial frames.

Compared with previous works making use of the threading perspective, we therefore establish a formal correspondence with the 
3+1 approach (i) by supplying the mentioned sets of equations, and (ii) by writing them with respect to an arbitrary four-dimensional 
vector basis and its dual. In preceding analyses, either the former or the latter item was not realized. 

We do not intend, however, to provide an exhaustive investigation of the subject, and we shall restrict ourselves to the derivation of 
the mentioned equations along with the associated material.\footnote{%
Our derivation and most of the notation we use are inspired from Gourgoulhon's book on the slicing formalism \cite{gourg:book}.
}~A unified and thorough description of the 1+3 and 3+1 techniques was supplied by Jantzen and collaborators in 
\cite{jantzen91,jantzen92,jantzen:book}. Notably, the 1+3 decomposition of the four-Riemann tensor (hence the Gauss, Codazzi and 
Ricci equations) can be found in \cite{jantzen:book}, but it is valid in a specific basis only (the one adapted to the congruence of curves). 
As mentioned earlier, we provide in the present paper the general formulation for this decomposition (and hence the general formulation 
for the threaded Einstein equations). In addition to allow for a transparent comparison with the 3+1 approach, such an extension is needed 
in order to work along a specific congruence (that involved in the threading of space--time) while permitting the space-like basis to move 
along any other congruence.\footnote{%
This feature will be used in forthcoming works on tilted cosmologies and cosmological deviation theory.
}\textsuperscript{,}\footnote{%
For completeness, let us mention that the threaded approach to the thermodynamics of a perfect fluid is given in \cite{gourg06}.
}

We proceed as follows. In Section \ref{sec:derivative_op}, we recall the definitions of the covariant derivative and of the Lie derivative. 
This is followed, in Section \ref{sec:fund_forms}, by the introduction of the two fundamental forms of the local spatial frames. Section 
\ref{sec:3connect} is devoted to the presentation of the spatial covariant derivative, while Section \ref{sec:3curv} focuses on the definition 
of the spatial curvature tensors. In Section \ref{sec:proj.curv_tens}, we provide in a first part the different projections of the four-Riemann 
curvature, which result in the Gauss, Codazzi and Ricci equations. The evolution equation for the spatial metric is also given. In a second 
part, we provide the different projections of the four-Ricci tensor. (All these relations are purely geometric and are valid independently 
of the gravitational theory.) We apply the splitting procedure to the Einstein field equation in Section \ref{sec:1+3.efe}. Finally, in Section 
\ref{sec:conclusion}, we summarize and discuss our results.\footnote{%
The main results of the paper are gathered in Propositions \ref{prop:gauss_rel} to \ref{prop:ricci_rel} and Corollaries \ref{cor:gauss_rel} 
to \ref{cor:ricci_rel} for the splitting of the manifold (`S' for splitting), and in Propositions \ref{prop:einstein-gauss} to 
\ref{prop:einstein-ricci} for the 1+3 form of Einstein's equation (`E' for Einstein).
}

This presentation is complemented by two appendices. Appendix \ref{app:kin_approach} provides the reformulation of the 1+3 Einstein 
equations in terms of the kinematical quantities of the fluid, and Appendix \ref{app:bases&coord} is dedicated to the construction of 
bases and coordinates adapted to the congruence.

\subsection{Notation and conventions}

We consider in what follows a smooth four-dimensional manifold $\scM$ endowed with the metric tensor $\bm g$ of signature 
$(-, +, +, +)$, and a set of time-like integral curves $\scC$ in $\scM$ described by the unit tangent vector field $\bm u$. Our analysis 
will be conducted locally at a generic point $p$ of $\scM$.

\subsubsection{Tangent, cotangent spaces and canonical isomorphism}

We denote by $\tangm$ and $\cotangm$, respectively, the four-dimensional spaces of vectors and 1-forms on $\scM$ at $p$. These 
spaces are mapped onto each other by means of the canonical isomorphism induced by the metric. We denote by a flat the isomorphism 
$\tangm \to \cotangm$ and by a sharp the reverse isomorphism $\cotangm \to \tangm$. Hence 
\begin{itemize}
	\item[$\bullet$] for any vector $\bm v$ in $\tangm$, $\bm v^\flat$ stands for the unique linear form in $\cotangm$ such that 
		\beq \label{eq:def.dual_form}
			\forall \, \bm w \in \tangm \quad\;\;\,
				\langle \bm v^\flat, \bm w \rangle := \bm v \cdot \bm w \, , 
		\eeq
	\item[$\bullet$] for any 1-form $\bm \form$ in $\cotangm$, $\bm \form^\sharp$ stands for the unique vector in $\tangm$ such that 
		\beq \label{eq:def.dual_vect}
			\forall \, \bm v \in \tangm \quad\;\;\,
				\bm \form^\sharp \cdot \bm v := \langle \bm \form, \bm v \rangle \, . 
		\eeq
\end{itemize}
This mapping is extended to multilinear forms as follows. For any tensor field $\bm T$ of type $(0,2)$ and any two vectors $\bm v$ 
and $\bm w$ on $\scM$, we denote 
\begin{itemize}
	\item[$\bullet$] by $\bm T^\sharp$ the tensor field of type $(1,1)$ such that 
		\beq \label{eq:def.dual_tensUpDown}
			\bm T^\sharp (\bm v^\flat, \bm w) := \bm T (\bm v, \bm w) \, , 
		\eeq
	\item[$\bullet$] by $\bm T^{\bdot \sharp}$ the tensor field of type $(1,1)$ such that 
		\beq \label{eq:def.dual_tensDownUp}
			\bm T^{\bdot \sharp} (\bm v, \bm w^\flat) := \bm T (\bm v, \bm w) \, , 
		\eeq
	\item[$\bullet$] by $\bm T^{\sharp \sharp}$ the tensor field of type $(2,0)$ such that 
		\beq \label{eq:def.dual_tensUpUp}
			\bm T^{\sharp \sharp} (\bm v^\flat, \bm w^\flat) := \bm T (\bm v, \bm w) \, . 
		\eeq
\end{itemize}
The mapping of forms of higher types is defined following the same prescription. 

Here and in the sequel, we employ a dot to indicate the scalar product of two vector fields taken with $\bm g$, 
\beqNo
	\forall \, (\bm v, \bm w) \in \tangm \times \tangm \quad\;\;
		\bm v \cdot \bm w := \bm g (\bm v, \bm w) \, ,
\eeqNo
and angle brackets to represent the action of linear forms on vector fields, 
\beqNo
	\forall \, (\bm \form, \bm v) \in \cotangm \times \tangm \quad\;\; 
		\langle \bm \form, \bm v \rangle := \bm \form (\bm v) \, . 
\eeqNo
%

\subsubsection{Spatial and temporal spaces}

The local spatial frames induced by the congruence of curves (orthogonal to $\bm u$) are collectively referred to as $\fframes$. 
They do not form any family of hypersurfaces in the present framework as we allow the congruence to manifest vorticity. 

%

We denote by $\tangf$ the three-dimensional space of \textit{spatial} vectors at $p$, such that 
\beqNo
	\forall \, \bm v \in \tangf \quad\;\;\,
		\bm v \cdot \bm u = 0 \, , 
\eeqNo
and by $\cotangf$ the three-dimensional space of \textit{spatial} 1-forms at $p$, such that 
\beqNo
	\forall \, \bm \form \in \cotangf \quad\;\;\,
		\langle \bm \form, \bm u \rangle = 0 \, . 
\eeqNo
The one-dimensional space of \textit{temporal} vectors is identified by $\mathrm{Vect}_p(\bm u)$, and its elements are such that 
\beqNo
	\forall \, \bm v \in \mathrm{Vect}_p(\bm u) \quad\, 
	\exists \lambda \in \bbR \quad\;
		\bm v = \lambda \bm u \, , 
\eeqNo
and the one-dimensional space of \textit{temporal} 1-forms is identified by $\mathrm{Vect}_p(\bm u^\flat)$, with 
\beqNo
	\forall \, \bm \form \in \mathrm{Vect}_p(\bm u^\flat) \quad\, 
	\exists \lambda \in \bbR \quad\;
		\bm \form = \lambda \bm u^\flat \, . 
\eeqNo
A tensor field on $\scM$ will be called spatial (resp.\ temporal) if it vanishes whenever one of its arguments is temporal (resp.\ spatial).

From the above definitions we write the orthogonal decomposition of the tangent space $\tangm$ as 
\beqNo
	\tangm
		= \tangf \oplus \mathrm{Vect}_p(\bm u) \, , 
\eeqNo
and that of the cotangent space $\cotangm$ according to 
\beqNo
	\cotangm
		= \cotangf \oplus \mathrm{Vect}_p(\bm u^\flat) \, . 
\eeqNo
%

\subsubsection{Bases and components}

We denote by $\{ \basis_\alpha \}$ an arbitrary basis of $\tangm$ and by $\{ \basis^\alpha \}$ the dual basis in $\cotangm$ satisfying 
by definition\footnote{%
Greek letters are assigned to four-dimensional counters and indices, they run in $\{ 0, 1, 2, 3 \}$. Latin letters refer to space-like counters 
and indices, running in $\{ 1, 2, 3 \}$. The objects $\{ \basis_i \}$ and $\{ \basis^i \}$ are space-like but not necessarily spatial. (For instance, 
for the $\{ \basis_i \}$, we have $\basis_i \cdot \basis_i > 0$ but not necessarily $\basis_i \cdot \bm u = 0$.) We adopt Einstein's summation 
convention over repeated letters.
}
\beqNo
	\langle \basis^\alpha, \basis_\beta \rangle 
		:= \kroneck{\alpha}{\beta} \, . 
\eeqNo
If not explicitly specified, and unless otherwise stated, the components of any tensor are written with respect to these bases. We have for 
a vector $\bm v$ and a 1-form $\bm \form$
\beqNo
	\bm v 
		= \left\langle \basis^\alpha, \bm v \right\rangle \basis_\alpha
		=: v^\alpha \basis_\alpha \, , \qquad 
	\bm \form 
		= \left\langle \bm \form, \basis_\alpha \right\rangle \basis^\alpha
		=: \form_\alpha \basis^\alpha \, , 
\eeqNo
and for a tensor $\bm T$ of type $(k,l)$
\begin{gather*}
	\bm T 
		=: T^{\alpha_1 \ldots \alpha_k}_{\phantom{\alpha_1 \ldots \alpha_k} \beta_1 \ldots \beta_l} 
		\, \basis_{\alpha_1} \otimes \ldots \otimes \basis_{\alpha_k} \otimes 
		\basis^{\beta_1} \otimes \ldots \otimes \basis^{\beta_l} \, , \\ 
	\mathrm{with} \quad 
	T^{\alpha_1 \ldots \alpha_k}_{\phantom{\alpha_1 \ldots \alpha_k} \beta_1 \ldots \beta_l} 
		:= \bm T (\basis^{\alpha_1}, \ldots, \basis^{\alpha_k}, \basis_{\beta_1}, \ldots, \basis_{\beta_l}) \, . 
\end{gather*}
In particular, we write the components of the metric as $g_\albe$ and those of its inverse as $g^\albe$, with 
\beqNo
	g^{\alpha \gamma} g_{\gamma \beta} 
		= \kroneck{\alpha}{\beta} \, .
\eeqNo
The components of the linear form $\bm v^\flat$ associated with $\bm v$ (cf.~Eq.~\eqref{eq:def.dual_form}) and those of the vector 
$\bm \form^\sharp$ associated with $\bm \form$ (cf.~Eq.~\eqref{eq:def.dual_vect}) are expressed with respect to the components of 
$\bm v$ and $\bm \form$ respectively as 
\beqNo
	(\bm v^\flat)_\alpha 
		=: v_\alpha
		= g_{\alpha \gamma} v^\gamma \, ,
	\qquad
	(\bm \form^\sharp)^\alpha 
		=: \form^\alpha
		= g^{\alpha \gamma} \form_\gamma \, ,
\eeqNo
and the components of the tensors $\bm T^\sharp$, $\bm T^{\bdot \sharp}$ and $\bm T^{\sharp \sharp}$ associated with $\bm T$ 
(cf.\ Eqs.~\eqref{eq:def.dual_tensUpDown}, \eqref{eq:def.dual_tensDownUp} and \eqref{eq:def.dual_tensUpUp}) are written in terms 
of those of $\bm T$ according to 
\begin{gather*}
	(\bm T^\sharp)^{\alpha}_{\phantom{\alpha} \beta}
		=: T^{\alpha}_{\phantom{\alpha} \beta} 
		= g^{\alpha \gamma} \, T_{\gamma \beta} \, , \qquad
	(\bm T^{\bdot \sharp})_{\alpha}^{\phantom{\alpha} \beta}
		=: T_{\alpha}^{\phantom{\alpha} \beta} 
		= g^{\beta \gamma} \, T_{\alpha \gamma} \, , \\
	(\bm T^{\sharp \sharp})^\albe
		=: T^\albe
		= g^{\alpha \gamma} g^{\beta \delta} \, T_{\gamma \delta} \, .
\end{gather*}
(We drop the musical symbols in the component notation for the sake of clarity. No confusion should arise.)

\section{Derivative operators \label{sec:derivative_op}}

\subsection{Covariant derivative}

We consider in this work the torsionless and metric-compatible connection $\connect$ on $\scM$. 

\subsubsection{Definition}

The covariant derivative constructs from a tensor field $\bm T$ of type $(k,l)$ a new tensor field $\connect \bm T$ of type $(k,l+1)$, 
the components of which are written 
\beqNo
	(\connect \bm T)^{\alpha_1 \ldots \alpha_k}_{\phantom{\alpha_1 \ldots \alpha_k} \beta_1 \ldots \beta_l \gamma} \, .
\eeqNo
The coefficients of the connection with respect to the bases $\{ \basis_\alpha \}$ and $\{ \basis^\alpha \}$ are defined by 
\beq \label{eq:def.connect_coef}
	\connect \basis_\alpha
		=: \coef{\gamma}{\delta}{\alpha} \, \basis_\gamma \otimes \basis^\delta 
	\quad \Leftrightarrow \quad
	\connect \basis^\alpha
		=: - \coef{\alpha}{\delta}{\gamma} \, \basis^\gamma \otimes \basis^\delta \, .
\eeq
\begin{remark}
	We follow Hawking and Ellis' convention \cite{hawking73} for the order of the covariant indices of the connection coefficients. 
	Note, however, that this choice does not affect the form of our main results.\footnote{%
	See Propositions \ref{prop:gauss_rel} to \ref{prop:ricci_rel}, Corollaries \ref{cor:gauss_rel} to \ref{cor:ricci_rel} and 
	Propositions \ref{prop:einstein-gauss} to \ref{prop:einstein-ricci}.
	}
\end{remark}
%
%
\smallskip

\noindent
The covariant derivative of $\bm T$ along a vector $\bm v$ defines a tensor field of the same type as $\bm T$; it is written 
\beqNo
	\connect_{\bm v} \bm T 
		:= \connect \bm T (\, \underbrace{\ldots}_{k+l} \, , \bm v) \, ,
\eeqNo\\[-4mm]
and in component form 
\beq \label{eq:comp.cd}
	v^\gamma (\connect \bm T)^{\alpha_1 \ldots \alpha_k}_{\phantom{\alpha_1 \ldots \alpha_k} \beta_1 \ldots \beta_l \gamma} \, .
\eeq
%

\subsubsection{Covariant derivative of a tensor}

The components of the covariant derivative of a tensor $\bm T$ of type $(k,l)$ are given by 
\begin{align} \label{eq:comp.cd_tens}
	T^{\alpha_1 \ldots \alpha_k}_{\phantom{\alpha_1 \ldots \alpha_k} \beta_1 \ldots \beta_l \, ; \, \gamma} = 
	& \;\; 
	\basis_\gamma 
		\left( T^{\alpha_1 \ldots \alpha_k}_{\phantom{\alpha_1 \ldots \alpha_k} \beta_1 \ldots \beta_l} \right) 
	+ \sum_{i = 1}^{k} 
		\, T^{\alpha_1 \ldots \!\! { { {\scriptstyle i \atop \downarrow} \atop \scriptstyle \delta} \atop } \!\! \ldots \alpha_k}_{%
			\phantom{\alpha_1 \ldots \delta \ldots \alpha_k} \beta_1 \ldots \beta_l} 
		\coef{\alpha_i}{\gamma}{\delta} \nonumber \\ 
	& \;\, 
	- \sum_{i = 1}^{l} 
		\, T^{\alpha_1 \ldots \alpha_k}_{\phantom{\alpha_1 \ldots \alpha_k} %
			\beta_1 \ldots \!\! { \atop {\scriptstyle \delta \atop {\uparrow \atop \scriptstyle i} } } \!\! \ldots \beta_l} 
		\coef{\delta}{\gamma}{\beta_i} \, ,
\end{align}
where we use henceforth the short-hand 
\beqNo
	T^{\alpha_1 \ldots \alpha_k}_{\phantom{\alpha_1 \ldots \alpha_k} \beta_1 \ldots \beta_l \, ; \, \gamma} 
		:= (\connect \bm T)^{\alpha_1 \ldots \alpha_k}_{\phantom{\alpha_1 \ldots \alpha_k} \beta_1 \ldots \beta_l \gamma} \, .
\eeqNo
For a function $f$ on $\scM$, they reduce to 
\beq \label{eq:comp.cd_function}
	f_{, \, \alpha} 
		:= (\connect f)_\alpha = \basis_\alpha (f) \, .
\eeq
\begin{remark}
	We use a semicolon for the component form of the covariant derivative of a tensor and a comma when it comes to a function.
\end{remark}
%

\subsubsection{Torsion tensor}

The torsion of the connection is defined by 
\beq \label{eq:def.torsion.I}
	\torsion (\bm v, \bm w) 
		:= \connect_{\bm v} \bm w - \connect_{\bm w} \bm v - [\bm v, \bm w] \, ,
\eeq
or, alternatively, by 
\beq \label{eq:def.torsion.II}
	\torsion (\bm v, \bm w) (f) 
		:= \connect \connect f (\bm v, \bm w) - \connect \connect f (\bm w, \bm v) \, , 
\eeq
with $f$ being a function and $\bm v$ and $\bm w$ two vectors on $\scM$. It vanishes identically since we consider the connection 
to be Levi-Civita.

\subsubsection{Curvature tensors}

The Riemann curvature tensor of $\scM$ is given by 
\beq \label{eq:def.riem}
	{}^{4\!}\bm{Riem} ( \bm \form, \bm w, \bm v_1, \bm v_2) 
		:= \langle \bm \form, 
		\connect_{\bm v_1} \connect_{\bm v_2} \bm w 
		- \connect_{\bm v_2} \connect_{\bm v_1} \bm w  
		- \connect_{[\bm v_1, \bm v_2]} \bm w \rangle \, , 
\eeq
with $\bm \form$ in $\cotangm$ and $\bm w$, $\bm v_1$ and $\bm v_2$ in $\tangm$, and the Ricci curvature tensor is defined as 
\beqNo
	{}^{4\!}\bm R(\bm v, \bm w)
		:={}^{4\!}\bm{Riem} (\basis^\alpha, \bm v, \basis_\alpha, \bm w) \, . 
\eeqNo
For convenience we denote the components of ${}^{4\!}\bm{Riem}$ by ${}^{4\!}R^{\alpha}{}_{\! \beta \mnu}$ rather than 
${}^{4\!}\mathit{Riem}^{\alpha} {\!}_{\beta \mnu}$. In the case of a torsion-free connection, we have 
\beq \label{eq:prop.riem}
	v^\mu_{\phantom{\mu} ; \, \beta \, ; \, \alpha} - v^\mu_{\phantom{\mu} ; \, \alpha \, ; \, \beta}
		= {}^{4\!}R^\mu_{\phantom{\mu} \nu \albe} v^\nu \, , 
\eeq
for any vector $\bm v$ on $\scM$.

\subsection{Lie derivative}

\subsubsection{Definition}

The Lie derivative of a tensor $\bm T$ of type $(k,l)$ along a vector $\bm v$ defines a new tensor $\CL_{\bm v} \bm T$ of type $(k,l)$, 
the components of which are written 
\beqNo
	\CL_{\bm v} T^{\alpha_1 \ldots \alpha_k}_{\phantom{\alpha_1 \ldots \alpha_k} \beta_1 \ldots \beta_l}
		:= \left( \CL_{\bm v} \bm T \right)^{\alpha_1 \ldots \alpha_k}_{\phantom{\alpha_1 \ldots \alpha_k} \beta_1 \ldots \beta_l} \, .
\eeqNo
%

\subsubsection{Lie derivative of a tensor}

For a torsion-free connection, the components of the Lie derivative of $\bm T$ along a vector $\bm v$ are given by 
\begin{align} \label{eq:comp.lie_tens}
	\CL_{\bm v} T^{\alpha_1 \ldots \alpha_k}_{\phantom{\alpha_1 \ldots \alpha_k} \beta_1 \ldots \beta_l} = 
	&\;\; 
	v^\gamma 
		\, T^{\alpha_1 \ldots \alpha_k}_{\phantom{\alpha_1 \ldots \alpha_k} \beta_1 \ldots \beta_l \, ; \, \gamma}
	- \sum_{i = 1}^{k} 
		\, T^{\alpha_1 \ldots \!\! { { {\scriptstyle i \atop \downarrow} \atop \scriptstyle \gamma} \atop } \!\! \ldots \alpha_k}_{%
			\phantom{\alpha_1 \ldots \gamma \ldots \alpha_k} \beta_1 \ldots \beta_l} 
		\, v^{\alpha_i}_{\phantom{\alpha_i} ; \, \gamma} \nonumber \\ 
	& \; 
	+ \sum_{i = 1}^{l} 
		\, T^{\alpha_1 \ldots \alpha_k}_{\phantom{\alpha_1 \ldots \alpha_k} %
			\beta_1 \ldots \!\! { \atop {\scriptstyle \gamma \atop {\uparrow \atop \scriptstyle i} } } \!\! \ldots \beta_l} 
		\, v^\gamma_{\phantom{\gamma} ; \, \beta_i} \, . 
\end{align}
For a function $f$ on $\scM$, they read 
\beqNo
	\CL_{\bm v} f 
		= v^\alpha f_{, \alpha} \, .
\eeqNo
%

\subsubsection{Lie bracket and structure coefficients \label{subsubsec:lie_brack}}

The Lie bracket of two vectors $\bm v$ and $\bm w$ is defined by 
\beq \label{eq:def.lie_brack}
	[ \bm v, \bm w ] (f)
		:= \bm v (\bm w (f)) - \bm w (\bm v (f)) \, . 
\eeq
The structure coefficients of the vector basis $\{ \basis_\alpha \}$ are defined from the Lie brackets of its elements: 
\beqNo
	[ \basis_\alpha, \basis_\beta ] (f)
		=: \struct{\gamma}{\alpha}{\beta} \, \basis_\gamma (f) \, ,
\eeqNo
and they are anti-symmetric in their two last indices. For a torsionless connection we have 
\beq \label{eq:rel.coef}
	\struct{\gamma}{\alpha}{\beta} 
		= 2 \, \coef{\gamma}{[\alpha}{\beta]} \, ,
\eeq
where the brackets indicate anti-symmetrization over the indices enclosed.

\section{Fundamental forms \label{sec:fund_forms}}

We introduce in this section the two fundamental forms of the local spatial frames of the congruence.

\subsection{First fundamental form}

\subsubsection{Orthogonal projector}

The orthogonal projector onto the spatial frames of $\scC$ is defined by 
\beq \label{eq:def.orth_proj.vect}
	\begin{array}{cccl}
		\orthproj \colon & \tangm & \to & \tangf \\[2pt] 
		& \bm v & \mapsto & \bm v + \langle \bm u^\flat, \bm v \rangle \, \bm u 
	\end{array}
\eeq
for vectors, and by 
\beq \label{eq:def.orth_proj.form}
	\begin{array}{cccl}
		\orthproj \colon & \cotangm & \to & \cotangf \\[2pt] 
		& \bm \form & \mapsto & \bm \form + \langle \bm \form, \bm u \rangle \, \bm u^\flat 
	\end{array}
\eeq
for 1-forms. Its components are given by 
\beq \label{eq:comp.orth_proj}
	( \orthproj )^\alpha_{\phantom{\alpha} \beta} 
		:= \proj{\alpha}{\beta}
		= \kroneck{\alpha}{\beta} + u^\alpha u_\beta \, . 
\eeq
We construct the spatial projection of tensors of higher types by requiring the fulfillment of the relation 
\beq \label{eq:prop.orth_proj}
	\orthproj (\bm T \otimes \bm S) 
		= (\orthproj \bm T) \otimes (\orthproj \bm S) \, , 
\eeq
for any two tensors $\bm T$ and $\bm S$. The components of the spatial projection $\orthproj \bm T$ of a tensor $\bm T$ of type 
$(k,l)$ are then given by 
\beqNo
	(\orthproj \bm T)^{\alpha_1 \ldots \alpha_k}_{\phantom{\alpha_1 \ldots \alpha_k} \beta_1 \ldots \beta_l} 
		= \proj{\alpha_1}{\mu_1} \ldots \proj{\alpha_k}{\mu_k} \, \proj{\nu_1}{\beta_1} \ldots \proj{\nu_l}{\beta_l}
		\, T^{\mu_1 \ldots \mu_k}_{\phantom{\mu_1 \ldots \mu_k} \nu_1 \ldots \nu_l} \, . 
\eeqNo
\begin{proof}
	We decompose the tensor $\bm T$ into 
	\beqNo
		\bm T 
			= T^{\alpha_1 \ldots \alpha_k}_{\phantom{\alpha_1 \ldots \alpha_k} \beta_1 \ldots \beta_l} 
			\, \basis_{\alpha_1} \otimes \ldots \otimes \basis_{\alpha_k} \otimes 
			\basis^{\beta_1} \otimes \ldots \otimes \basis^{\beta_l} \, .
	\eeqNo
	Applying recursively \eqref{eq:prop.orth_proj} on the projected expression, and noticing with the help of \eqref{eq:def.orth_proj.vect} 
	and \eqref{eq:comp.orth_proj} that 
	\beqNo
		\orthproj (\basis_\alpha) 
			= \proj{\beta}{\alpha} \basis_\beta \, , 
	\eeqNo
	we obtain the sought-after expression.
\end{proof}
%

\subsubsection{Metric on the spatial frames}

The orthogonal projector $\orthproj$ introduced by Eqs.~\eqref{eq:def.orth_proj.vect}, \eqref{eq:def.orth_proj.form} and 
\eqref{eq:comp.orth_proj} defines the first fundamental form of the spatial frames 
\beqNo
	\bm \gamma
		:= \bm g + \bm u^\flat \otimes \bm u^\flat \, , 
\eeqNo
and in component form 
\beq \label{eq:comp.3metric}
	\gamma_\albe 
		= g_\albe + u_\alpha u_\beta \, . 
\eeq
When its domain of definition is restricted to $\tangf \times \tangf$, the symmetric, non-degenerate, bilinear form $\bm \gamma$ 
plays the role of the (Riemannian) spatial metric on the local spatial frames.

\subsection{Second fundamental form}

\subsubsection{Weingarten map}

The Weingarten map (or shape operator) of the spatial frames associates to a spatial vector the covariant derivative of the flow vector 
$\bm u$ along that vector: 
\beqNo
	\begin{array}{cccl}
		\bm \chi \colon & \tangf & \to & \tangf \\[2pt] 
		& \bm v & \mapsto & \connect_{\bm v} \bm u \, . 
	\end{array}
\eeqNo
The image of $\tangf$ under $\bm \chi$ is indeed in $\tangf$ as $\bm u$ is unitary.

%
%
\smallskip
\begin{remark}
	From the torsion-free character of the connection together with \eqref{eq:def.torsion.I}, we find 
	\beq \label{eq:prop.weing_map}
		\bm v \cdot \bm \chi (\bm w) 
			= \bm w \cdot \bm \chi (\bm v) + \bm u \cdot [ \bm v, \bm w ] \, , 
	\eeq
	for any spatial vectors $\bm v$ and $\bm w$. Because the collection of spatial frames does not form any hypersurfaces, the Lie 
	bracket $[ \bm v, \bm w ]$ is not spatial (cf.\ Frobenius' theorem in, e.g., \cite{wald84}). The term $\bm u \cdot [\bm v, \bm w]$ 
	therefore does not cancel and the shape operator is not self-adjoint.
\end{remark}
%

\subsubsection{Extrinsic curvature}

The second fundamental form of the spatial frames is defined for two spatial vectors by 
\beqNo
	\begin{array}{cccl}
		\bm k \colon & \tangf \times \tangf & \to & \bbR \\[2pt] 
		& ( \bm v, \bm w ) & \mapsto & - \bm v \cdot \bm \chi (\bm w) \, . 
	\end{array}
\eeqNo
We extend this definition to arbitrary vectors on $\scM$ upon writing 
\beqNo
	\begin{array}{cccl}
		\bm k \colon & \tangm \times \tangm & \to & \bbR \\[2pt] 
		& ( \bm v, \bm w ) & \mapsto & - \orthproj (\bm v) \cdot \bm \chi (\orthproj (\bm w)) \, . 
	\end{array}
\eeqNo
With the definition of the shape operator we can formulate the extrinsic curvature (or deformation tensor) as 
\beq \label{eq:ext_curv}
	\bm k ( \bm v, \bm w ) 
		= - \orthproj (\bm v) \cdot \connect_{\orthproj (\bm w)} \bm u \, . 
\eeq
\begin{remark}
	Because of the non self-adjointness of the Weingarten map, the extrinsic curvature fails to be symmetric. We have indeed, from 
	Eqs.~\eqref{eq:prop.weing_map} and \eqref{eq:ext_curv}, 
	\beq \label{eq:prop.ext_curv.I}
		\bm k ( \bm v, \bm w ) - \bm k ( \bm w, \bm v ) 
			= - \bm u \cdot [ \bm v, \bm w ] \, ,
	\eeq
	for any spatial vectors $\bm v$ and $\bm w$.
\end{remark}
%

\subsubsection{Relation between \texorpdfstring{$\bm k$ and $\connect \bm u^\flat$}%
	{the extrinsic curvature and the covariant derivative of the flow 1-form}}

Making use of Eq.~\eqref{eq:ext_curv} along with Eq.~\eqref{eq:def.orth_proj.vect} we deduce, for two arbitrary vectors $\bm v$ and 
$\bm w$, 
\beqNo
	\bm k ( \bm v, \bm w ) 
		= - \connect \bm u^\flat ( \bm v, \bm w ) 
		- \langle \bm a^\flat, \bm v \rangle \langle \bm u^\flat, \bm w \rangle \, , 
\eeqNo
where $\bm a := \connect_{\bm u} \bm u$ stands for the curvature vector of the congruence. (Note that $\bm a$ is spatial.) This expression can be 
equivalently written 
\beq \label{eq:prop.ext_curv.II}
	\connect \bm u^\flat 
		= - \bm k - \bm a^\flat \otimes \bm u^\flat \, , 
\eeq
and in component form 
\beq \label{eq:comp.prop.ext_curv.II}
	u_{\alpha \, ; \, \beta} 
		= - k_\albe - a_\alpha u_\beta \, . 
\eeq
\begin{remark}
	In Appendix \ref{app:kin_approach}, we decompose $\connect \bm u^\flat$ into the kinematical quantities of the fluid filling the 
	space--time, the world lines of which are identified with the integral curves of $\scC$. Expressing the extrinsic curvature in terms 
	of the fluid expansion rate, shear and vorticity, we there provide the kinematical formulation of the 1+3 Einstein equations.
\end{remark}
%

\section{Spatial connection \label{sec:3connect}}

This section is devoted to the presentation of the spatial connection associated with the connection $\connect$ on $\scM$.

\subsection{Introduction}

\subsubsection{Definition}

We define the connection $\connectf$ on the local spatial frames of the congruence by 
\beq \label{eq:def.3connect}
	\connectf \bm T 
		:= \orthproj (\connect \bm T) \, ,
\eeq
for any \textit{spatial} tensor $\bm T$.\footnote{%
We do not extend this definition to non-spatial tensors, as for such objects the operator $\connectf$ would lose its character of derivative.
}~The spatial covariant derivative constructs from a spatial tensor of type $(k,l)$ a new spatial tensor of type $(k,l+1)$, the components of 
which satisfy the relation 
\beq \label{eq:comp.3connect}
	\left( \connectf \bm T \right)^{\alpha_1 \ldots \alpha_k}_{\phantom{\alpha_1 \ldots \alpha_k} \beta_1 \ldots \beta_l \gamma} 
		= \, \proj{\alpha_1}{\mu_1} \ldots \proj{\alpha_k}{\mu_k} \proj{\nu_1}{\beta_1} \ldots \proj{\nu_l}{\beta_l} 
		\proj{\delta}{\gamma} \left( \connect \bm T \right)^{\mu_1 \ldots \mu_k}_{\phantom{\mu_1 \ldots \mu_k} 
		\nu_1 \ldots \nu_l \delta} \, . 
\eeq
\begin{remark}
	Although the spatial connection only applies to spatial tensors, the object $\connectf \bm T$ can accept as arguments vectors and 
	1-forms not necessarily spatial ($\connectf \bm T$ is a tensor defined on $\scM$). Accordingly we can write four-dimensional 
	components for this term.
\end{remark}
%
%
\smallskip

\noindent
The spatial covariant derivative of $\bm T$ along an arbitrary vector $\bm v$ defines a spatial tensor field of the same type as $\bm T$. 
It is written 
\beqNo
	\connectf_{\bm v} \bm T 
		:= \connectf \bm T (\, \underbrace{\ldots}_{k+l} \, , \bm v) \, ,
\eeqNo\\[-4mm]
and in component form 
\beq \label{eq:comp.3dir_der}
	v^\gamma 
		\left( \connectf \bm T \right)^{\alpha_1 \ldots \alpha_k}_{\phantom{\alpha_1 \ldots \alpha_k} %
		\beta_1 \ldots \beta_l \gamma} \, .
\eeq
%

\subsubsection{Spatial covariant derivative of a tensor}

The components of the spatial covariant derivative of a spatial tensor $\bm T$ of type $(k,l)$ are given by 
\begin{align} \label{eq:comp.scd_tens}
	T^{\alpha_1 \ldots \alpha_k}_{\phantom{\alpha_1 \ldots \alpha_k} \beta_1 \ldots \beta_l \, || \, \gamma} 
		& = \, \proj{\alpha_1}{\mu_1} \ldots \proj{\nu_l}{\beta_l} \proj{\delta}{\gamma} \, \basis_\delta 
			\Big( T^{\mu_1 \ldots \mu_k}_{\phantom{\mu_1 \ldots \mu_k} \nu_1 \ldots \nu_l} \Big) \\ 
		& \!
		+ \sum_{i = 1}^{k} 
			\, T^{\alpha_1 \ldots \!\! { { {\scriptstyle i \atop \downarrow} \atop \scriptstyle \delta} \atop } \!\! \ldots \alpha_k}_{%
				\phantom{\alpha_1 \ldots \delta \ldots \alpha_k} \beta_1 \ldots \beta_l} 
			\, \coeff{\alpha_i}{\gamma}{\delta} 
		- \sum_{i = 1}^{l} 
			\, T^{\alpha_1 \ldots \alpha_k}_{\phantom{\alpha_1 \ldots \alpha_k} %
				\beta_1 \ldots \!\! { \atop {\scriptstyle \delta \atop {\uparrow \atop \scriptstyle i} } } \!\! \ldots \beta_l} 
			\, \coeff{\delta}{\gamma}{\beta_i} \, , \nonumber 
\end{align}
where we have used the short-hand 
\beqNo
	T^{\alpha_1 \ldots \alpha_k}_{\phantom{\alpha_1 \ldots \alpha_k} \beta_1 \ldots \beta_l \, || \, \gamma}
		:= \left( \connectf \bm T \right)^{\alpha_1 \ldots \alpha_k}_{\phantom{\alpha_1 \ldots \alpha_k} 
			\beta_1 \ldots \beta_l \gamma} 
\eeqNo
and defined the coefficients of the spatial connection as 
\beq \label{eq:def.3connect_coef}
	\coeff{\gamma}{\alpha}{\beta} 
		:= \proj{\gamma}{\delta} \proj{\rho}{\alpha} \proj{\sigma}{\beta} \, \coef{\delta}{\rho}{\sigma} \, . 
\eeq
For a function $f$ on $\scM$ they read 
\beq \label{eq:comp.scd_function}
	f_{\, | \, \alpha} 
		= \proj{\beta}{\alpha} \, f_{, \, \beta} 
		= \proj{\beta}{\alpha} \basis_\beta (f) \, . 
\eeq
\begin{proof}
	Equations \eqref{eq:comp.scd_tens} and \eqref{eq:comp.scd_function} are respectively obtained from \eqref{eq:comp.cd_tens} 
	and \eqref{eq:comp.cd_function} and upon using \eqref{eq:comp.3connect}.
\end{proof}
\begin{remark}
	We use two strokes for the component form of the spatial covariant derivative of a tensor and only one when it comes to a function.
\end{remark}
%

%
%
\smallskip
\begin{remark}
	The coefficients of the spatial connection are spatial in the sense that any contraction with $u^\alpha$ or $u_\alpha$ vanishes.
\end{remark}
%

\subsection{Properties}

\subsubsection{Relations between \texorpdfstring{$\connectf$ and $\connect$}{the spatial and the manifold connections}}

For a spatial tensor $\bm T$ and an arbitrary vector $\bm v$ on $\scM$, we have 
\beq \label{eq:prop.3connect.I}
	\connectf_{\bm v} \bm T 
		= \connectf_{\orthproj (\bm v)} \bm T 
		= \orthproj \big( \hspace{1pt} \connect_{\orthproj (\bm v)} \bm T \hspace{1pt} \big) \, . 
\eeq
\begin{proof}
	Using the spatial character of $\connectf \bm T$, we reformulate the components of $\connectf_{\bm v} \bm T$ given by 
	Eq.~\eqref{eq:comp.3dir_der} into 
	\beqNo
		v^\gamma \proj{\delta}{\gamma} 
			\left( \connectf \bm T \right)^{\alpha_1 \ldots \alpha_k}_{\phantom{\alpha_1 \ldots \alpha_k} 
			\beta_1 \ldots \beta_l \delta} \, . 
	\eeqNo
	Noticing with \eqref{eq:def.orth_proj.vect} and \eqref{eq:comp.orth_proj} that the terms $v^\gamma \gamma^\delta {\!}_\gamma$ 
	denote the components of $\orthproj (\bm v)$, we deduce the first equality. The second equality is obtained by applying 
	\eqref{eq:comp.3connect} and \eqref{eq:comp.cd} to the above expression.
\end{proof}

For two spatial vectors $\bm v$ and $\bm w$, we have the relation 
\beq \label{eq:prop.3connect.II}
	{\bf D}_{\bm v} \bm w
		= \connect_{\bm v} \bm w + \bm k (\bm w, \bm v) \, \bm u \, . 
\eeq
\begin{proof}
	The components of $\connect_{\bm v} \bm w$ are decomposed orthogonally according to 
	\begin{align*}
		v^\beta w^\alpha_{\phantom{\alpha} ; \, \beta}
			& = v^\beta \, \kroneck{\alpha}{\gamma} \kroneck{\delta}{\beta} 
				w^\gamma_{\phantom{\gamma} ; \, \delta} \\ 
			& = v^\beta \, \proj{\alpha}{\gamma} \proj{\delta}{\beta} w^\gamma_{\phantom{\gamma} ; \, \delta}
				- v^\delta u^\alpha u_\gamma w^\gamma_{\phantom{\gamma} ; \, \delta} \\
			& = v^\beta w^\alpha_{\phantom{\alpha} || \, \beta} 
				+ v^\delta u^\alpha w^\gamma u_{\gamma \, ; \, \delta} \, ,
	\end{align*}
	where we have used Eq.~\eqref{eq:comp.orth_proj} and the spatial character of $\bm v$ for the second equality, and 
	Eq.~\eqref{eq:comp.3connect} and the spatial character of $\bm w$ for the third equality. Inserting 
	Eq.~\eqref{eq:comp.prop.ext_curv.II} into the last line we get 
	\beqNo
		v^\beta w^\alpha_{\phantom{\alpha} ; \, \beta} 
			= v^\beta w^\alpha_{\phantom{\alpha} || \, \beta} - u^\alpha w^\gamma v^\delta k_{\gamma \delta} \, , 
	\eeqNo
	which, written in tensor form, concludes the proof.
\end{proof}

\noindent
Similarly, for a spatial vector $\bm v$ and a spatial 1-form $\bm \form$ we can deduce 
\beqNo
	{\bf D}_{\bm v} \bm \form 
		= \connect_{\bm v} \bm \form + \bm k^\sharp (\bm \form, \bm v) \, \bm u^\flat \, . 
\eeqNo
%

\subsubsection{Compatibility with the spatial metric}

The components of $\connectf \bm \gamma$ are written by means of \eqref{eq:comp.3connect} and \eqref{eq:comp.3metric} as 
\beqNo
	\gamma_{\albe \, || \, \gamma} 
		= \proj{\rho}{\alpha} \proj{\sigma}{\beta} \proj{\delta}{\gamma} \, g_{\rho \sigma \, ; \, \delta} \, . 
\eeqNo
From the compatibility of $\connect$ with respect to $\bm g$ we obtain 
\beqNo
	\connectf \bm \gamma = \bm 0 \, . 
\eeqNo
The spatial connection is therefore compatible with the spatial metric.

\subsubsection{Torsion tensor}

Following Eq.~\eqref{eq:def.torsion.I}, we define the torsion tensor of the spatial connection as 
\beq \label{eq:def.3torsion.I}
	\torsionf (\bm v, \bm w) 
		:= \connectf_{\bm v} \bm w - \connectf_{\bm w} \bm v - [ \bm v, \bm w ] \, , 
\eeq
for any spatial vectors $\bm v$ and $\bm w$. We extend this definition to arbitrary vectors on $\scM$ upon writing 
\beq \label{eq:def.3torsion.I.ext}
	\torsionf (\bm v, \bm w) 
		:= \connectf_{\orthproj (\bm v)} \orthproj (\bm w) - \connectf_{\orthproj (\bm w)} \orthproj (\bm v)
		- \big[ \orthproj (\bm v), \orthproj (\bm w) \big] \, . 
\eeq
Developing the right-hand side with the help of \eqref{eq:prop.3connect.II}, and using the torsion-free character of the manifold 
connection along with \eqref{eq:def.torsion.I}, we obtain 
\beq \label{eq:exp.3torsion}
	\torsionf (\bm v, \bm w) 
		= \big( \bm k (\bm w, \bm v) - \bm k (\bm v, \bm w) \big) \, \bm u \, ,  
\eeq
for any vectors  $\bm v$ and $\bm w$. The torsion of $\connectf$ is generated by the anti-symmetric part of the extrinsic curvature 
tensor or, equivalently, by the temporal part of the Lie bracket of two spatial vectors (cf.~Eq.~\eqref{eq:prop.ext_curv.I}). It is a tensor 
of type $(1,2)$, with a temporal contravariant part and a spatial covariant part.

%
%
\smallskip
\begin{remark}
	In the 1+3 kinematical formulation, the torsion tensor of the spatial connection is induced by the fluid vorticity (see Appendix 
	\ref{app:kin_approach}).
\end{remark}

\noindent
Considering the analogue of \eqref{eq:def.torsion.II} for the definition of the torsion, 
\beq \label{eq:def.3torsion.II}
	\torsionf (\bm v, \bm w) (f) 
		:= \connectf \connectf f (\bm v, \bm w) -  \connectf  \connectf f (\bm w, \bm v) \, , 
\eeq
with $f$ being a function and $\bm v$ and $\bm w$ two vectors on $\scM$, also yields \eqref{eq:exp.3torsion} as demonstrated in 
the following proof.

\begin{proof}
	Definition \eqref{eq:def.3torsion.II} reads in component form 
	\beq \label{eq:dev.3torsion}
		v^\alpha w^\beta \, \torscoeff{\gamma}{\alpha}{\beta} \, \basis_\gamma (f) 
			= v^\alpha w^\beta \left( f_{\, | \, \alpha \, || \, \beta} - f_{\, | \, \beta \, || \, \alpha} \right) \, . 
	\eeq
	From \eqref{eq:comp.3connect} we have 
	\beqNo
		f_{\, | \, \alpha \, || \, \beta}
			= \proj{\gamma}{\alpha} \proj{\delta}{\beta} \left( \proj{\lambda}{\gamma} \, f_{, \, \lambda} \right)_{; \, \delta} \, . 
	\eeqNo
	Expanding the right-hand side and making use of \eqref{eq:comp.orth_proj}, we write 
	\beqNo
		f_{\, | \, \alpha \, || \, \beta} 
			= \proj{\gamma}{\alpha} \proj{\delta}{\beta} \, f_{, \gamma \, ; \, \delta} 
			+ \proj{\gamma}{\alpha} \proj{\delta}{\beta} \, u^\lambda f_{, \, \lambda} \, u_{\gamma \, ; \, \delta} \, . 
	\eeqNo
	With the help of \eqref{eq:comp.prop.ext_curv.II} we then deduce 
	\beqNo
		f_{\, | \, \alpha \, || \, \beta} - f_{\, | \, \beta \, || \, \alpha} 
			= \proj{\gamma}{\alpha} \proj{\delta}{\beta} 
				\left( f_{, \gamma \, ; \, \delta} - f_{, \, \delta \, ; \, \gamma} \right) 
			- \proj{\gamma}{\alpha} \proj{\delta}{\beta} \, u^\lambda f_{, \, \lambda} 
				\left( k_{\gamma \delta} - k_{\delta \gamma} \right) \, . 
	\eeqNo
	Invoking the torsion-free character of the connection $\connect$ to cancel the first term of the sum and the spatial character 
	of the extrinsic curvature to reformulate the second, we get 
	\beqNo
		f_{\, | \, \alpha \, || \, \beta} - f_{\, | \, \beta \, || \, \alpha} 
			= \left( k_{\beta \alpha} - k_\albe \right) u^\lambda f_{, \, \lambda} \, . 
	\eeqNo
	Inserting this expression back into \eqref{eq:dev.3torsion} and using \eqref{eq:comp.cd_function}, we conclude the proof.
\end{proof}
\begin{remark}
	The torsion tensor is defined in \cite{boersma94a,boersma94b,boersma94c} and \cite{jantzen:book} according to 
	\beq \label{eq:def.3torsion.alt}
		\torsionf (\bm v, \bm w) 
			:= \connectf_{\bm v} \bm w - \connectf_{\bm w} \bm v - \orthproj([ \bm v, \bm w ]) \, . 
	\eeq
	(It vanishes for a Levi-Civita manifold connection.) This expression involves only the spatial part of $[\bm v, \bm w]$ in comparison 
	with \eqref{eq:def.3torsion.I}. Although it does not enter \eqref{eq:def.3torsion.alt}, the temporal part of $[\bm v, \bm w]$ plays 
	the same role in these works as in ours. It only manifests itself in a different way, and more specifically through the so-called deficiency 
	term in \cite{boersma94a,boersma94b,boersma94c} and in an explicit form in \cite{jantzen:book}. As a consequence, choosing one or 
	the other definition for the torsion does not affect the relations deduced from the threading of the manifold.
	
	Note, however, that only our approach is consistent with the alternate definition \eqref{eq:def.3torsion.II}. This implies that by picking 
	definition \eqref{eq:def.3torsion.alt}, the commutation of two spatial covariant derivatives successively applied to a function induces 
	a term that is \textit{not} source of torsion.
\end{remark}
%

\section{Spatial curvature tensors \label{sec:3curv}}

We introduce in this section the Riemann tensor, Ricci tensor and Ricci scalar of the spatial frames.

\subsection{Spatial Riemann curvature}

Following the definition of the Riemann curvature tensor of $\scM$ (cf.\ Eq.~\eqref{eq:def.riem}), one may want to define the spatial 
Riemann tensor according to 
\beq \label{eq:def.3riem.torsion-free}
	\bm{Riem} ( \bm \form, \bm w, \bm v_1, \bm v_2) 
		:= \langle \bm \form, 
		\connectf_{\bm v_1} \connectf_{\bm v_2} \bm w 
		- \connectf_{\bm v_2} \connectf_{\bm v_1} \bm w  
		- \connectf_{[\bm v_1, \bm v_2]} \bm w \rangle \, , 
\eeq
with $\bm \form$ in $\cotangf$ and $\bm v_1$, $\bm v_2$ and $\bm w$ in $\tangf$.\footnote{%
Note that we have, from Eq.~\eqref{eq:prop.3connect.I}, 
$\connectf_{[\bm v_1, \bm v_2]} \bm w %
	= \connectf_{\orthproj ([\bm v_1, \bm v_2])} \bm w %
	= \orthproj (\connect_{\orthproj([\bm v_1, \bm v_2])} \bm w)$. 
\label{foo:3riem}
}~However, as it was noticed for instance in \cite{boersma94a}, the object hence defined is not a tensor because of its lack of linearity. 
We have indeed 
\beqNo
	\bm{Riem} ( \bm \form, f \bm w, \bm v_1, \bm v_2) 
		= f \bm{Riem} ( \bm \form, \bm w, \bm v_1, \bm v_2) 
		- \bm u \cdot [\bm v_1, \bm v_2] \; \bm u (f) \, \langle \bm \form, \bm w \rangle \, , 
\eeqNo
for any function $f$ on $\scM$.

\begin{proof}
	From Eq.~\eqref{eq:def.3riem.torsion-free} we have 
	\begin{align} \label{eq:dev.3riem.torsion-free.lin.I}
		\bm{Riem} ( \bm \form, f \bm w, \bm v_1, \bm v_2) = 
			& \; f \bm{Riem} ( \bm \form, \bm w, \bm v_1, \bm v_2) \nonumber \\
			& - \left( \connectf_{\bm v_1} \connectf_{\bm v_2} f 
			- \connectf_{\bm v_2} \connectf_{\bm v_1} f 
			- \connectf_{[\bm v_1, \bm v_2]} f \right) \langle \bm \form, \bm w \rangle \, .
	\end{align}
	Let us start by working on the first two terms between parentheses. For any spatial vector $\bm v$ we write, by means of 
	Eqs.~\eqref{eq:comp.scd_function} and \eqref{eq:comp.cd_function}, 
	\beq \label{eq:dev.3riem.torsion-free.lin.II}
		\connectf_{\bm v} f 
			= v^\alpha \, f_{| \, \alpha} 
			= v^\alpha \, \proj{\beta}{\alpha} \, f_{, \, \beta}
			= v^\alpha \, f_{, \, \alpha}
			= \bm v (f) \, , 
	\eeq
	and thus we infer 
	\beqNo
		\connectf_{\bm v_1} \connectf_{\bm v_2} f - \connectf_{\bm v_2} \connectf_{\bm v_1} f 
			= [\bm v_1, \bm v_2] (f) \, .
	\eeqNo
	The last term between parentheses, on the other hand, can be cast into 
	\begin{align*}
		\connectf_{[\bm v_1, \bm v_2]} f 
			& = [\bm v_1, \bm v_2]^\alpha \, \proj{\beta}{\alpha} \, f_{, \, \beta} \\
			& = [\bm v_1, \bm v_2]^\alpha \, f_{, \, \alpha}
				+ u_\alpha \, [\bm v_1, \bm v_2]^\alpha \, u^\beta \, f_{, \, \beta} \\[2pt]
			& = [\bm v_1, \bm v_2] (f) 
				+ \bm u \cdot [\bm v_1, \bm v_2] \; \bm u (f) \, . 
	\end{align*}
	Subtracting this expression from the one above and plugging the result into \eqref{eq:dev.3riem.torsion-free.lin.I}, we 
	conclude the proof.
\end{proof}

\noindent
The non-linearity of the object \eqref{eq:def.3riem.torsion-free} stems from the fact that the temporal part of the Lie bracket 
$[\bm v_1, \bm v_2]$ does not vanish or, equivalently, from the fact that the extrinsic curvature of the spatial frames is not symmetric 
(cf.~Eq.~\eqref{eq:prop.ext_curv.I}). We present in what follows a definition that circumvents the issue.

%
%
\smallskip
\begin{remark}
	In the 1+3 kinematical formulation, the non-linearity of \eqref{eq:def.3riem.torsion-free} is caused by the vorticity of the fluid.
\end{remark}
%

\subsubsection{Definition}

The Riemann curvature tensor of the local spatial frames can be defined as 
\begin{align} \label{eq:def.3riem}
	\begin{array}{cl}
		\bm{Riem} \colon & \cotangf \times \tangf \times \tangf \times \tangf \; \to \; \bbR \\[4pt]
		& ( \bm \form, \bm w, \bm v_1, \bm v_2)
		\; \mapsto \; 
		\Big\langle \bm \form, 
			\connectf_{\bm v_1} \connectf_{\bm v_2} \bm w 
			- \connectf_{\bm v_2} \connectf_{\bm v_1} \bm w  
			- \orthproj \left( \connect_{[\bm v_1, \bm v_2]} \bm w \right) \! \Big\rangle \, . 
	\end{array}
\end{align}
\begin{remark}
	This definition was introduced originally by Massa \cite{massa74a,massa74b,massa74c} and studied later by Jantzen and collaborators 
	\cite{jantzen92,jantzen:book} and by Boersma and Dray \cite{boersma94a}. In their works, Jantzen \textit{et al.}\ derive the expressions 
	of the Gauss, Codazzi and Ricci relations in a particular basis only (the one adapted to the congruence). On the other hand, Boersma 
	and Dray derive, in tensorial form, the Gauss relation only, and they provide the component form of \eqref{eq:def.3riem} in a specific 
	basis only. In the present article we provide, as we shall see below, (i) the components of the spatial Riemann tensor in an arbitrary 
	four-dimensional basis and its dual, and (ii) the different projections of the Riemann tensor of $\scM$ onto the congruence and the 
	spatial frames.\footnote{%
	These projections result in the Gauss, Codazzi and Ricci equations (cf.\ respectively Propositions \ref{prop:gauss_rel}, 
	\ref{prop:cod_rel} and \ref{prop:ricci_rel}). Note that we supplement these relations with the equation for the variation of the 
	spatial metric along the integral curves in Proposition \ref{prop:evol.3met}.
	}
\end{remark}
%

%
%
\smallskip
\begin{remark}
	In the terminology of Jantzen \textit{et al.}\ \cite{jantzen92}, Eq.~\eqref{eq:def.3riem} stands for the `Fermi--Walker spatial curvature'. 
	Another definition for the spatial Riemann tensor (historically the first) was given by Zel'manov in \cite{zelm56,zelm59} (see also 
	\cite{ferrarese63,ferrarese65,massa74c,jantzen91}). For an analysis of the latter definition (`Lie spatial curvature') along with the proposal 
	for yet another definition (`co-rotating Fermi--Walker curvature'), we refer the reader to \cite{jantzen91,jantzen:book}.\footnote{%
	In the present work we focus solely on \eqref{eq:def.3riem} as it is the expression that most closely resembles 
	\eqref{eq:def.3riem.torsion-free} (see also footnote \ref{foo:3riem}).
	}
\end{remark}
%

%
%
\smallskip
\noindent
We extend the definition of the spatial Riemann curvature \eqref{eq:def.3riem} to arbitrary tensors on $\scM$ upon writing 
\begin{align} \label{eq:def.3riem.ext}
	\bm{Riem} ( \bm \form, \bm w, \bm v_1, \bm v_2) :=
		& \; \Big\langle \orthproj (\bm \form), \,
			\connectf_{\orthproj (\bm v_1)} \, \connectf_{\orthproj (\bm v_2)} \, \orthproj (\bm w) 
			- \connectf_{\orthproj (\bm v_2)} \, \connectf_{\orthproj (\bm v_1)} \, \orthproj (\bm w) \nonumber \\
		& \;\;\, - \orthproj \Big( \connect_{ \left[ \orthproj (\bm v_1), \orthproj (\bm v_2) \right]} \orthproj (\bm w) \Big) 
		\Big\rangle \, . 
\end{align}
Its components (with respect to an arbitrary vector basis and its dual) are given by\footnote{%
For simplicity, we denote the components of $\bm{Riem}$ by $R^{\alpha} {}_{\beta \mnu}$ rather than 
$\mathit{Riem}^{\alpha} {\!}_{\beta \mnu}$.
}
\begin{align} \label{eq:comp.3riem}
	R^\alpha_{\phantom{\alpha} \beta \mnu} 
		= \proj{\alpha}{\gamma} \proj{\delta}{\beta} \proj{\rho}{\mu} \proj{\sigma}{\nu} 
		\; \bigg\{ 
			& \basis_\rho \left( \coeff{\gamma}{\sigma}{\delta} \right) 
			- \basis_\sigma \left( \coeff{\gamma}{\rho}{\delta} \right) 
			+ \coeff{\lambda}{\sigma}{\delta} \, \coeff{\gamma}{\rho}{\lambda} \nonumber \\[1pt]
			& \hspace{-5mm} 
			- \coeff{\lambda}{\rho}{\delta} \, \coeff{\gamma}{\sigma}{\lambda}
			- \structf{\lambda}{\rho}{\sigma} \coeff{\gamma}{\lambda}{\delta} \nonumber \\
			& \hspace{-5mm} 
			+ \torscoeff{\lambda}{\rho}{\sigma} \, \coef{\gamma}{\lambda}{\delta} 
			+ \basis_\rho (u_\delta) \, \basis_\sigma (u^\gamma)
			- \basis_\rho (u^\gamma) \, \basis_\sigma (u_\delta) 
		\bigg\} \, , 
\end{align}
where we have defined the spatial structure coefficients $\structf{}{}{}$ according to 
\beq \label{eq:def.3struct_coef}
	\structf{\gamma}{\alpha}{\beta} 
		:= \proj{\gamma}{\delta} \proj{\rho}{\alpha} \proj{\sigma}{\beta} \, \struct{\delta}{\rho}{\sigma} 
		=  2 \, \coeff{\gamma}{[ \alpha}{\beta ]} \, . 
\eeq
(The structure coefficients $C^\delta {}_{\rho \sigma}$ were defined in Section \ref{subsubsec:lie_brack}, and we have used Eqs.~\eqref{eq:rel.coef} 
and \eqref{eq:def.3connect_coef} for the second equality.)

\begin{proof}
	Let us first define the set of spatial \textit{starry vectors} $\{ \sbasis{\alpha} \}$ by 
	\beq \label{eq:def.star_vect}
		\sbasis{\alpha} 
			:= \orthproj (\basis_\alpha) 
			= \basis_\alpha + u_\alpha \bm u \, ,
	\eeq
	and the set of spatial \textit{starry 1-forms} $\{ \scobasis{\alpha} \}$ as 
	\beq \label{eq:def.star_form}
		\scobasis{\alpha} 
			:= \orthproj (\basis^\alpha) 
			= \basis^\alpha + u^\alpha \bm u^\flat \, .
	\eeq
	For the purpose of the proof we list some of the properties fulfilled by these objects. We have for the vectors 
	\beq \label{eq:prop.star_vect}
		\sbasis{\alpha} 
			= \proj{\beta}{\alpha} \basis_\beta \, , \qquad 
		\sbasis{\alpha} 
			= \proj{\beta}{\alpha} \sbasis{\beta} \, , \qquad 
		u^\alpha \sbasis{\alpha} 
			= 0 \, ,
	\eeq
	and, similarly, for the 1-forms 
	\beq \label{eq:prop.star_form}
		\scobasis{\alpha} 
			= \proj{\alpha}{\beta} \basis^\beta \, , \qquad 
		\scobasis{\alpha} 
			= \proj{\alpha}{\beta} \scobasis{\beta} \, , \qquad 
		u_\alpha \scobasis{\alpha} = 0 \, . 
	\eeq
	In addition, the starry sets satisfy the relation 
	\beq \label{eq:prop.star_sets}
		\langle \scobasis{\alpha}, \sbasis{\beta} \rangle 
			= \proj{\alpha}{\beta} \, . 
	\eeq
	With the help of \eqref{eq:def.star_vect} and \eqref{eq:def.star_form} we can write the components of the spatial Riemann tensor 
	\eqref{eq:def.3riem.ext} in the form 
	\beq \label{eq:dev.comp.3riem.I}
		R^\alpha_{\phantom{\alpha} \beta \mnu} 
			= \left\langle \scobasis{\alpha}, \,
			\connectf_{\sbasis{\mu}} \connectf_{\sbasis{\nu}} \sbasis{\beta}
			- \connectf_{\sbasis{\nu}} \connectf_{\sbasis{\mu}} \sbasis{\beta} 
			- \orthproj \big( \connect_{[\sbasis{\mu}, \, \sbasis{\nu}]} \, \sbasis{\beta} \big) \right\rangle \, .
	\eeq
	For the sake of clarity, we divide the rest of the proof into the following four parts.
	
	\medskip\smallskip
	\noindent
	\textit{Part I.}
	Let us begin with the decomposition of the term $\connectf_{\sbasis{\beta}} \sbasis{\alpha}$. By means of the property 
	\beqNo
		\connect (f \bm T) 
			= f \connect \bm T + \bm T \otimes \connect f \, , 
	\eeqNo
	with $f$ being a function and $\bm T$ a tensor on $\scM$, together with Eqs.~\eqref{eq:def.3connect} and \eqref{eq:def.star_vect}, 
	we have 
	\beqNo
		\connectf \sbasis{\alpha} 
			= \orthproj \Big( \connect \basis_\alpha 
			+ u_\alpha \connect \bm u 
			+ \bm u \otimes \connect u_\alpha \Big) \, . 
	\eeqNo
	Using Eq.~\eqref{eq:def.connect_coef} for the first term between parentheses, Eq.~\eqref{eq:comp.cd_tens} applied to $\bm u$ 
	for the second term and Eqs.~\eqref{eq:prop.orth_proj}, \eqref{eq:def.star_vect} and \eqref{eq:def.star_form} on the resulting 
	expression, we find 
	\beqNo
		\connectf \sbasis{\alpha}
			 = \left( \coef{\gamma}{\delta}{\alpha} 
			+ u_\alpha \, \basis_\delta (u^\gamma) 
			+ u_\alpha \coef{\gamma}{\delta}{\lambda} u^\lambda \right) 
				\sbasis{\gamma} \otimes \scobasis{\delta} \, . 
	\eeqNo
	From Eqs.~\eqref{eq:comp.orth_proj}, \eqref{eq:prop.star_vect}, \eqref{eq:prop.star_form} and \eqref{eq:def.3connect_coef}, 
	this expression becomes 
	\beqNo
		\connectf \sbasis{\alpha}
			= \Big( \coeff{\gamma}{\delta}{\alpha} + u_\alpha \proj{\gamma}{\lambda} \sbasis{\delta} (u^\lambda) \Big) 
			\, \sbasis{\gamma} \otimes \scobasis{\delta} \, . 
	\eeqNo
	At last, making use of Eq.~\eqref{eq:prop.star_sets}, we obtain 
	\beq \label{eq:dev.comp.3riem.II}
		\connectf_{\sbasis{\beta}} \sbasis{\alpha} 
			= \Big( \coeff{\gamma}{\beta}{\alpha} + u_\alpha \proj{\gamma}{\delta} \sbasis{\beta} (u^\delta) \Big) 
				\, \sbasis{\gamma} \, .
	\eeq
	The components of $\connectf_{\sbasis{\beta}} \sbasis{\alpha}$ with respect to an arbitrary vector basis follow from 
	Eq.~\eqref{eq:def.star_vect}. They are identical to those given in the above decomposition. We shall however consider the form 
	\eqref{eq:dev.comp.3riem.II} in what follows as we will apply another spatial covariant derivative to 
	$\connectf_{\sbasis{\beta}} \sbasis{\alpha}$  (recall that this operator is only defined for spatial tensors).
	
	\medskip\smallskip
	\noindent
	\textit{Part II.}
	We now turn to the decomposition of the first term of the right-hand side of \eqref{eq:dev.comp.3riem.I}. From 
	\eqref{eq:dev.comp.3riem.II} we have 
	\begin{align*}
		\connectf_{\sbasis{\mu}} \connectf_{\sbasis{\nu}} \sbasis{\beta} = 
			& \; \left( \coeff{\gamma}{\nu}{\beta} + u_\beta \proj{\gamma}{\delta} \sbasis{\nu} (u^\delta) \right) 
				\connectf_{\sbasis{\mu}} \sbasis{\gamma} \\
			& \; + \connectf_{\sbasis{\mu}} 
				\!\left( \coeff{\gamma}{\nu}{\beta} + u_\beta \proj{\gamma}{\delta} \sbasis{\nu} (u^\delta) \right) 
				\sbasis{\gamma} \, . 
	\end{align*}
	The first part of the sum is developed by using again Eq.~\eqref{eq:dev.comp.3riem.II} and the second part is expanded with the 
	help of Eqs.~\eqref{eq:dev.3riem.torsion-free.lin.II} and \eqref{eq:comp.orth_proj}. We reformulate the outcome by considering 
	the spatial character of $\coeff{}{}{}$ and the unitarity of $\bm u$, and we obtain 
	\begin{align*}
		\connectf_{\sbasis{\mu}} \connectf_{\sbasis{\nu}} \sbasis{\beta} 
			& = \bigg\{ 
				\sbasis{\mu} \left( \coeff{\gamma}{\nu}{\beta} \right) 
				- u_\beta u^\delta \, \sbasis{\nu} \left( \coeff{\gamma}{\mu}{\delta} \right)
				+ \coeff{\lambda}{\nu}{\beta} \, \coeff{\gamma}{\mu}{\lambda} \nonumber \\[-1pt]
			& + \proj{\gamma}{\delta} \, \sbasis{\mu} ( u_\beta) \, \sbasis{\nu} (u^\delta)
				- u_\beta u^\delta \, \sbasis{\mu} ( u^\gamma) \, \sbasis{\nu} (u_\delta)
				+ u_\beta \proj{\gamma}{\delta} \, \sbasis{\mu} \left( \sbasis{\nu} (u^\delta) \right)
			\bigg\} \; \sbasis{\gamma} \, . 
	\end{align*}
	From this expression and by means of Eqs.~\eqref{eq:def.lie_brack}, \eqref{eq:comp.orth_proj} and \eqref{eq:prop.star_vect}, 
	we write the first two terms of the right-hand side of \eqref{eq:dev.comp.3riem.I} in the form  
	\begin{align} \label{eq:dev.comp.3riem.III}
		\connectf_{\sbasis{\mu}} \connectf_{\sbasis{\nu}} \sbasis{\beta}
			- \connectf_{\sbasis{\nu}} \connectf_{\sbasis{\mu}} \sbasis{\beta} =
		& \; \bigg\{ 
		\proj{\delta}{\beta} \, \sbasis{\mu} \left( \coeff{\gamma}{\nu}{\delta} \right) 
		- \proj{\delta}{\beta} \, \sbasis{\nu} \left( \coeff{\gamma}{\mu}{\delta} \right) \nonumber \\[1pt]
		& \quad 
		+ \coeff{\lambda}{\nu}{\beta} \, \coeff{\gamma}{\mu}{\lambda}
		- \coeff{\lambda}{\mu}{\beta} \, \coeff{\gamma}{\nu}{\lambda} \nonumber \\[4pt]
		& \quad
		+ \proj{\delta}{\beta} \proj{\gamma}{\lambda} \, \sbasis{\mu} ( u_\delta) \, \sbasis{\nu} (u^\lambda)
		- \proj{\delta}{\beta} \proj{\gamma}{\lambda} \, \sbasis{\mu} (u^\lambda) \, \sbasis{\nu} ( u_\delta) \nonumber \\[1pt]
		& \quad 
		+ u_\beta \proj{\gamma}{\lambda} \left[ \sbasis{\mu}, \sbasis{\nu} \right] \! (u^\lambda)  
		\bigg\} \; \sbasis{\gamma} \, .
	\end{align}

	\noindent
	\textit{Interlude.}
	In order to proceed with the third part of the proof, we provide the expression of the Lie bracket of two starry vectors. 
	From Eqs.~\eqref{eq:def.3torsion.I.ext} and \eqref{eq:dev.comp.3riem.II} we have 
	\beq \label{eq:lie.star_vect}
		[ \sbasis{\alpha}, \sbasis{\beta} ] = 
			\, \left( \structf{\gamma}{\alpha}{\beta} 
				+ u_\beta \, \proj{\gamma}{\delta} \, \sbasis{\alpha} (u^\delta) 
				- u_\alpha \, \proj{\gamma}{\delta} \, \sbasis{\beta} (u^\delta) \right) \sbasis{\gamma} 
			- \torscoeff{\gamma}{\alpha}{\beta} \, \basis_\gamma \, , 
	\eeq
	where the coefficients $\structf{}{}{}$ are defined by \eqref{eq:def.3struct_coef}. 
	
	\medskip\smallskip
	\noindent
	\textit{Part III.}
	We now decompose the last term of the right-hand side of \eqref{eq:dev.comp.3riem.I}. Using \eqref{eq:lie.star_vect} and 
	\eqref{eq:def.star_vect} we write 
	\begin{align*}
		\connect_{[ \sbasis{\mu}, \, \sbasis{\nu} ]} \, \sbasis{\beta} = 
			& \; \Big( 
				\structf{\delta}{\mu}{\nu} 
				+ u_\nu \proj{\delta}{\lambda} \sbasis{\mu} (u^\lambda) 
				- u_\mu \proj{\delta}{\lambda} \sbasis{\nu} (u^\lambda) 
			\Big) \, \connect_{\sbasis{\delta}} \sbasis{\beta} \\
			& \, - \torscoeff{\lambda}{\mu}{\nu} \, \Big( 
				\connect_{\basis_\lambda} \basis_\beta 
				+ u_\beta \connect_{\basis_\lambda} \bm u 
				+ \basis_\lambda (u_\beta) \, \bm u 
			\Big) \, . 
	\end{align*}
	The second expression in parentheses is developed with the help of Eq.~\eqref{eq:def.connect_coef}, Eq.~\eqref{eq:comp.cd_tens} 
	written for $\bm u$ and Eq.~\eqref{eq:comp.orth_proj}. We get 
	\begin{align*}
		\connect_{[ \sbasis{\mu}, \, \sbasis{\nu} ]} \, \sbasis{\beta} = 
			& \; \Big( \structf{\delta}{\mu}{\nu} 
				+ u_\nu \proj{\delta}{\lambda} \sbasis{\mu} (u^\lambda) 
				- u_\mu \proj{\delta}{\lambda} \sbasis{\nu} (u^\lambda) \Big) 
				\, \connect_{\sbasis{\delta}} \sbasis{\beta} \\
			& \, - \torscoeff{\lambda}{\mu}{\nu} \left( 
				\proj{\delta}{\beta} \coef{\gamma}{\lambda}{\delta} \, \basis_\gamma 
				+ u_\beta \, \basis_\lambda (u^\gamma) \, \basis_\gamma 
				+ \basis_\lambda (u_\beta) \, \bm u \right) . 
	\end{align*}
	We project this equality onto the spatial frames and we apply Eqs.~\eqref{eq:prop.3connect.I} and \eqref{eq:dev.comp.3riem.II} 
	on the first line and Eq.~\eqref{eq:def.star_vect} on the second. Expanding the resulting expression and making use of the spatial 
	character of $\coeff{}{}{}$, we end up with 
	\begin{align} \label{eq:dev.comp.3riem.IV}
		\orthproj \left( \connect_{[\sbasis{\mu}, \, \sbasis{\nu}]} \, \sbasis{\beta} \right) = 
			& \; \bigg\{ \!
				- \proj{\delta}{\beta} \, u^\sigma u_\nu \, \sbasis{\mu} \left( \coeff{\gamma}{\sigma}{\delta} \right) 
				+ \proj{\delta}{\beta} \, u^\rho u_\mu \, \sbasis{\nu} \left( \coeff{\gamma}{\rho}{\delta} \right) \nonumber \\[1pt]
				& \;\;\; 
				+ \structf{\delta}{\mu}{\nu} \, \coeff{\gamma}{\delta}{\beta} 
				+ u_\beta \proj{\gamma}{\lambda} \structf{\delta}{\mu}{\nu} \, \sbasis{\delta}(u^\lambda) \nonumber \\[5pt]
				& \;\;\;
				+ u_\beta u_\nu \proj{\gamma}{\lambda} \sbasis{\mu} (u^\delta) \, \sbasis{\delta} (u^\lambda)
				- u_\beta u_\mu \proj{\gamma}{\lambda}  \sbasis{\nu} (u^\delta) \, \sbasis{\delta} (u^\lambda) \nonumber \\[1pt]
				& \;\;\; 
				- \torscoeff{\lambda}{\mu}{\nu} \! \left( 
					\proj{\delta}{\beta} \coef{\gamma}{\lambda}{\delta} + u_\beta \basis_\lambda (u^\gamma) \right) 
			\bigg\} \, \sbasis{\gamma} \, . 
	\end{align}

	\noindent
	\textit{Part IV.}
	Subtracting \eqref{eq:dev.comp.3riem.IV} from \eqref{eq:dev.comp.3riem.III} and making use of \eqref{eq:prop.star_vect}, 
	we obtain 
	\begin{align*}
		& \connectf_{\sbasis{\mu}} \connectf_{\sbasis{\nu}} \sbasis{\beta}
			- \connectf_{\sbasis{\nu}} \connectf_{\sbasis{\mu}} \sbasis{\beta} 
			- \orthproj \left( \connect_{[\sbasis{\mu}, \, \sbasis{\nu}]} \, \sbasis{\beta} \right) \\ 
		& 
		= \proj{\delta}{\beta} \proj{\rho}{\mu} \proj{\sigma}{\nu} \; \bigg\{ 
			\basis_\rho \left( \coeff{\gamma}{\sigma}{\delta} \right) 
			- \basis_\sigma \left( \coeff{\gamma}{\rho}{\delta} \right) 
			+ \coeff{\lambda}{\sigma}{\delta} \, \coeff{\gamma}{\rho}{\lambda} \\[1pt]
			& \hspace{2.3cm} 
			- \coeff{\lambda}{\rho}{\delta} \, \coeff{\gamma}{\sigma}{\lambda}
			- \structf{\lambda}{\rho}{\sigma} \coeff{\gamma}{\lambda}{\delta} 
			+ \torscoeff{\lambda}{\rho}{\sigma} \, \coef{\gamma}{\lambda}{\delta} \\[1pt]
			& \hspace{2.3cm} 
			+ \basis_\rho (u_\delta) \, \basis_\sigma (u^\gamma)
			- \basis_\rho (u^\gamma) \, \basis_\sigma (u_\delta) 
		\bigg\} \; \sbasis{\gamma} \, . 
	\end{align*}
	Finally, inserting this expression into \eqref{eq:dev.comp.3riem.I} and using \eqref{eq:prop.star_sets}, we conclude the proof.
\end{proof}
%

\subsubsection{Properties}

The spatial Riemann curvature satisfies the relations 
\beq \label{eq:prop.3riem.I}
	R_{\alpha \beta \mnu} 
		= - R_{\alpha \beta \nu \mu} \, , \qquad
	R_{\albe \mnu} 
		= - R_{\beta \alpha \mnu} \, , \qquad
	R_{\alpha [\beta \mnu]}
		= 2 k_{\alpha [ \beta} \, k_{\mnu]} \, . 
\eeq
\begin{proof}
	From the symmetries of the four-Riemann tensor, 
	\beqNo
		{}^{4\!}R_{\alpha \beta \mnu} 
			= - {}^{4\!}R_{\alpha \beta \nu \mu} \, , \qquad
		{}^{4\!}R_{\albe \mnu} 
			= - {}^{4\!}R_{\beta \alpha \mnu} \, , \qquad
		{}^{4\!}R_{\alpha [\beta \mnu]}
			= 0 \, , 
	\eeqNo
	together with Proposition~\ref{prop:gauss_rel} (given hereafter), we obtain the above expressions.
\end{proof}

\noindent
From the properties \eqref{eq:prop.3riem.I}, we infer 
\beq \label{eq:prop.3riem.II}
	R_{\alpha \beta \mu \nu} - R_{\mnu \alpha \beta} 
		= k_{\mu \alpha} k_{\nu \beta} - k_{\mu \beta} k_{\nu \alpha}
		- k_{\alpha \mu} k_{\beta \nu} + k_{\beta \mu} k_{\alpha \nu} \, . 
\eeq
\begin{remark}
	When the extrinsic curvature is symmetric, we recover the usual first Bianchi identity for the spatial Riemann tensor (cf.\ 
	Eq.~\eqref{eq:prop.3riem.I}) and relation \eqref{eq:prop.3riem.II} vanishes.
\end{remark}
%

\subsection{Spatial Ricci tensor}

\subsubsection{Definition}

We define the Ricci curvature tensor of the spatial frames from the trace of the spatial Riemann tensor \eqref{eq:def.3riem.ext} taken 
on the first and third arguments: 
\beq \label{eq:def.3ricci}
	\begin{array}{cccl}
		\bm R \colon & \tangm \times \tangm & \to & \bbR \\[3pt]
		& (\bm v, \bm w)
		& \mapsto 
		& \bm{Riem} (\basis^\alpha, \bm v, \basis_\alpha, \bm w) \, . 
	\end{array}
\eeq
\begin{remark}
	The other traces of the Riemann curvature of $\scC$ either vanish or are equal (possibly up to a sign) to the one here defined. 
	(This can be shown from \eqref{eq:prop.3riem.I}.)
\end{remark}
%
%
\smallskip

\noindent
The spatial Ricci scalar is defined from the trace of the spatial Ricci tensor. It is written 
\beqNo
	R := \bm R^\sharp (\basis^\alpha,\basis_\alpha) \, . 
\eeqNo
%

\subsubsection{Property}

The spatial Ricci tensor satisfies the relation 
\beq \label{eq:prop.3ricci}
	R_{\alpha \beta} - R_{\beta \alpha} 
		= - 2 k k_{[ \alpha \beta ]} - 2 k^\gamma {}_{[ \alpha} k_{\beta ] \gamma} \, , 
\eeq
where $k$ is the trace of the extrinsic curvature tensor.

\begin{proof}
	This expression is obtained by taking the trace of Eq.~\eqref{eq:prop.3riem.II} on the first and third indices and by using 
	definition \eqref{eq:def.3ricci}.
\end{proof}
\begin{remark}
	Because the extrinsic curvature is not symmetric, the spatial Ricci tensor is not symmetric either.
\end{remark}
%

%
%
\vspace{1pt}
\begin{remark}
	In the 1+3 kinematical formulation, the spatial Ricci tensor fails to be symmetric on account of the fluid vorticity.
\end{remark}
%

\section{\texorpdfstring{$1+3$}{1+3} form of the curvature tensors \label{sec:proj.curv_tens}}

In this section, we provide the 1+3 formulation of the curvature tensors of the manifold in terms of the spatial curvature tensors 
introduced above.

The projections of the Riemann tensor of $\scM$, in terms of the Riemann tensor associated with $\scC$, are given in the first subsection. 
They yield the 1+3 version of the Gauss, Codazzi and Ricci equations. From these expressions, we supply in the ensuing part the different 
projections of the four-Ricci tensor in terms of the spatial Ricci tensor.

Our results are written in the form of propositions and corollaries. We here only add a few remarks; a thorough discussion follow in 
Section \ref{sec:conclusion}.

\subsection{\texorpdfstring{$1+3$}{1+3} form of the Riemann tensor} \label{subsec:proj.riem}

\subsubsection{Gauss equation}

\begin{propositionS} \label{prop:gauss_rel}
	In the 1+3 formalism the Gauss relation is given by 
	\beq \label{eq:gauss_rel.I}
		\proj{\mu}{\rho} \proj{\sigma}{\nu} \proj{\gamma}{\alpha} \proj{\delta}{\beta} 
			\, {}^{4\!}R^\rho_{\phantom{\rho} \sigma \gamma \delta} 
		=  R^\mu_{\phantom{\mu} \nu \albe} 
			+ k^\mu_{\phantom{\mu} \alpha} k_{\nu \beta} - k^\mu_{\phantom{\mu} \beta} k_{\nu \alpha} \, .
	\eeq
\end{propositionS}
\begin{proof}
	From Eq.~\eqref{eq:comp.scd_tens} we write the components of $\connectf \connectf \bm v$, with $\bm v$ being a spatial 
	vector, as 
	\begin{align*}
		v^\mu_{\phantom{\mu} || \, \beta \, || \, \alpha} 
			= \, \proj{\mu}{\gamma} \proj{\rho}{\alpha} \proj{\sigma}{\beta} \, \bigg\{ 
			& \sbasis{\rho} \left( \coeff{\gamma}{\sigma}{\delta} \right) v^\delta 
				+ \coeff{\lambda}{\sigma}{\delta} \, \coeff{\gamma}{\rho}{\lambda} \, v^\delta \\
			& \hspace{-8pt} - \coeff{\lambda}{\rho}{\sigma} \, \coeff{\gamma}{\lambda}{\delta} \, v^\delta
				- \coeff{\lambda}{\rho}{\sigma} \, \sbasis{\lambda} (v^\gamma) 
				+ \coeff{\gamma}{\sigma}{\lambda} \sbasis{\rho} (v^\lambda) \\[1pt]
			& \hspace{-8pt} + \coeff{\gamma}{\rho}{\lambda} \sbasis{\sigma} (v^\lambda) 
				+ \sbasis{\rho} \left( \sbasis{\sigma} (v^\gamma) \right) 
				+ u_\delta \, \sbasis{\rho} (u^\gamma) \, \sbasis{\sigma} (v^\delta) 
			\bigg\} \, . 
	\end{align*}
	From Eqs.~\eqref{eq:def.3struct_coef} and \eqref{eq:lie.star_vect} and by means of the spatial character of $\bm v$, we then obtain 
	\begin{align*}
		v^\mu_{\phantom{\mu} || \, \beta \, || \, \alpha} - v^\mu_{\phantom{\mu} || \, \alpha \, || \, \beta} 
			= \, \proj{\mu}{\gamma} \proj{\rho}{\alpha} \proj{\sigma}{\beta} \, \bigg\{ 
			& \proj{\delta}{\nu} \, \bigg(
				\sbasis{\rho} \left( \coeff{\gamma}{\sigma}{\delta} \right) 
				- \sbasis{\sigma} \left( \coeff{\gamma}{\rho}{\delta} \right) \\
			& + \coeff{\lambda}{\sigma}{\delta} \, \coeff{\gamma}{\rho}{\lambda} 
				- \coeff{\lambda}{\rho}{\delta} \, \coeff{\gamma}{\sigma}{\lambda} \\[4pt]
			& - \structf{\lambda}{\rho}{\sigma} \, \coeff{\gamma}{\lambda}{\delta} 
				+ \sbasis{\rho} (u_\delta) \, \sbasis{\sigma} (u^\gamma) \\
			&	- \sbasis{\rho} (u^\gamma) \, \sbasis{\sigma} (u_\delta) 
				\bigg) \, v^\nu
			- \torscoeff{\lambda}{\rho}{\sigma} \basis_\lambda (v^\gamma)
				\bigg\} \, . 
	\end{align*}
	Using Eq.~\eqref{eq:prop.star_vect} to remove the stars, Eq.~\eqref{eq:comp.3riem} to introduce the spatial Riemann curvature 
	and Eq.~\eqref{eq:comp.cd_tens} to insert the covariant derivative of $\bm v$, we find 
	\begin{align} \label{eq:dev.gauss_rel}
		v^\mu_{\phantom{\mu} || \, \beta \, || \, \alpha} - v^\mu_{\phantom{\mu} || \, \alpha \, || \, \beta} 
			= R^\mu_{\phantom{\mu} \nu \albe} v^\nu 
			- \proj{\mu}{\rho} \, \torscoeff{\lambda}{\alpha}{\beta} \, v^\rho_{\phantom{\rho} ; \, \lambda} \, . 
	\end{align}
	This relation concludes the first part of the proof.
	
	Equation \eqref{eq:dev.gauss_rel} strongly resembles the usual expression relating a spatial connection with torsion to its curvature 
	tensor. The fact that the torsion of the spatial connection has a non-spatial contravariant part prevents the use of the spatial connection 
	on the right-hand side of the expression.
	
	We now express the left-hand side of \eqref{eq:dev.gauss_rel} in terms of the covariant derivative of the manifold. From 
	Eqs.~\eqref{eq:comp.3connect} and \eqref{eq:comp.prop.ext_curv.II} we have 
	\begin{align*}
		v^\mu_{\phantom{\mu} || \, \beta \, || \, \alpha} 
			= \proj{\mu}{\rho} \proj{\gamma}{\alpha} \proj{\delta}{\beta} \, v^\rho_{\phantom{\rho} ; \, \delta \, ; \, \gamma} 
			- \proj{\mu}{\rho} \, k_{\beta \alpha} u^\lambda v^\rho_{\phantom{\rho} ; \, \lambda}
			- k^\mu_{\phantom{\mu} \alpha} k_{\lambda \beta} v^\lambda \, . 
	\end{align*}
	With Eq.~\eqref{eq:prop.riem} and the spatial character of $\bm v$ we write 
	\begin{align*}
		v^\mu_{\phantom{\mu} || \, \beta \, || \, \alpha} - v^\mu_{\phantom{\mu} || \, \alpha \, || \, \beta} = 
			& \; \proj{\mu}{\rho} \proj{\sigma}{\nu} \proj{\gamma}{\alpha} \proj{\delta}{\beta} 
				\, {}^{4\!}R^\rho_{\phantom{\rho} \sigma \gamma \delta} \, v^\nu 
				+ \proj{\mu}{\rho} \, \big( k_\albe - k_{\beta \alpha} \big) \, u^\lambda \, v^\rho_{\phantom{\rho} ; \lambda} \\[1.5pt]
			& - \big( k^\mu_{\phantom{\mu} \alpha} k_{\lambda \beta} 
				- k^\mu_{\phantom{\mu} \beta} k_{\lambda \alpha} \big) \, v^\lambda \, .
	\end{align*}
	Inserting this expression into \eqref{eq:dev.gauss_rel} and making use of \eqref{eq:exp.3torsion} in component form we deduce 
	\beqNo
		\proj{\mu}{\rho} \proj{\sigma}{\nu} \proj{\gamma}{\alpha} \proj{\delta}{\beta} 
			{}^{4\!}R^\rho_{\phantom{\rho} \sigma \gamma \delta} v^\nu
			= \big( R^\mu_{\phantom{\gamma} \nu \albe} 
			+ k^\mu_{\phantom{\mu} \alpha} k_{\nu \beta} - k^\mu_{\phantom{\mu} \beta} k_{\nu \alpha} \big) \, v^\nu \, .
	\eeqNo
	Noticing at last that this relation can be written not only for spatial vectors but also for any vectors on $\scM$ (thanks to the 
	presence of the orthogonal projector and the fact that both $\bm{Riem}$ and $\bm k$ are spatial), we conclude the proof.
\end{proof}
%

\subsubsection{Codazzi equation}

\begin{propositionS} \label{prop:cod_rel}
	In the 1+3 formalism the Codazzi relation is given by 
	\beq \label{eq:cod_rel.I}
		\proj{\mu}{\rho} u^\sigma \proj{\gamma}{\alpha} \proj{\delta}{\beta} 
			\, {}^{4\!}R^\rho_{\phantom{\rho} \sigma \gamma \delta} 
		= k^\mu_{\phantom{\mu} \alpha \, || \, \beta} - k^\mu_{\phantom{\mu} \beta \, || \, \alpha}
			- a^\mu \left( k_\albe - k_{\beta \alpha} \right) \, . 
	\eeq
\end{propositionS}
\begin{proof}
	This equation is obtained by fully projecting \eqref{eq:prop.riem} written for $\bm u$ onto the spatial frames, and by using 
	\eqref{eq:comp.prop.ext_curv.II} together with \eqref{eq:comp.3connect}.
\end{proof}
\begin{remark} \label{rem:cod}
	The Codazzi relation \eqref{eq:cod_rel.I} is anti-symmetric on the indices $\alpha$ and $\beta$.
\end{remark}
%

\subsubsection{Variation of the spatial metric}

We now search for the evolution equation of the spatial metric along the congruence. (This will be of use in the sequel for the derivation 
of the Ricci equation and its reformulation.) We introduce to this aim the temporal \textit{evolution vector} $\bm m$, such that 
\beqNo
	\bm m := \shl \bm u \, .
\eeqNo
$\shl$ stands for the (positive) threading lapse function; it allows for the choice of the pace of evolution along the integral curves of 
$\scC$. (Further information about $\shl$ can be found in Appendix \ref{app:bases&coord}.)

%
%
\medskip
\begin{remark}
	Even in the case where we regard spatial tensor fields we work within a four-dimensional manifold. Hence there is no `evolution' 
	of spatial quantities per se, but rather a variation along a preferred time-like direction.
\end{remark}
%
%
\smallskip
\begin{propositionS} \label{prop:evol.3met}
	The evolution of the spatial metric along the congruence of curves is given by 
	\beq \label{eq:evol.3met.cov}
		\CL_{\bm m} \gamma_\albe 
			= - \shl \big( k_\albe + k_{\beta \alpha} \big) 
	\eeq
	for its covariant components, 
	\beq \label{eq:evol.3met.contrav}
		\CL_{\bm m} \gamma^\albe =
			\shl \, \Big( k^\albe + k^{\beta \alpha} \Big) 
			+ \shl u^\alpha \Big( a^\beta - \frac{D^\beta \shl}{\shl} \Big) 
			+ \shl u^\beta \Big(  a^\alpha - \frac{D^\alpha \shl}{\shl} \Big) 
	\eeq
	for its contravariant components, and 
	\beq \label{eq:evol.3met.mixed}
		\CL_{\bm m} \gamma^\alpha_{\phantom{\alpha} \beta}
			= \shl u^\alpha \Big( a_\beta - \frac{D_\beta \shl}{\shl} \Big) 
	\eeq
	for its mixed components.
\end{propositionS}
\begin{proof}
	Each of these expressions is obtained by means of Eqs.~\eqref{eq:comp.lie_tens}, \eqref{eq:comp.3metric} and 
	\eqref{eq:comp.prop.ext_curv.II}.
\end{proof}
\begin{remark}
	No particular expression can be written for the curvature vector when working in arbitrary bases, on the contrary to the slicing 
	approach (see, e.g., \cite{gourg:book}). In Appendix \ref{app:bases&coord}, we provide an expression in a coordinate basis adapted 
	to the congruence.
\end{remark}
%

\subsubsection{Ricci equation}

\begin{propositionS} \label{prop:ricci_rel}
	In the 1+3 formalism the Ricci equation is given by 
	\beq \label{eq:ricci_rel.I}
		\gamma_{\alpha \rho} u^\sigma \proj{\gamma}{\beta} u^\delta 
				\, {}^{4\!}R^\rho_{\phantom{\rho} \sigma \gamma \delta}
		= \frac{1}{\shl} \CL_{\bm m} k_\albe
				+ k_{\gamma \alpha} k^\gamma_{\phantom{\gamma} \beta} 
				+ a_{\alpha \, || \, \beta} + a_\alpha a_\beta \, .
	\eeq
\end{propositionS}
\begin{proof}
	Projecting Eq.~\eqref{eq:prop.riem} written for $\bm u$ twice onto the spatial frames and once onto the congruence, we obtain 
	with the help of Eq.~\eqref{eq:comp.prop.ext_curv.II}
	\beq \label{eq:dev.evol.ext_curv.I}
		\gamma_{\alpha \rho} u^\sigma \proj{\gamma}{\beta} u^\delta 
			\, {}^{4\!}R^\rho_{\phantom{\rho} \sigma \gamma \delta} 
		= \gamma_{\alpha \rho} \proj{\gamma}{\beta} \, u^\delta \, k^\rho_{\phantom{\rho} \gamma \, ; \, \delta} 
			- k_{\alpha \gamma} k^\gamma_{\phantom{\gamma} \beta} 
			+ a_{\alpha \, || \, \beta} + a_\alpha a_\beta \, .
	\eeq
	On the other hand, we decompose from Eq.~\eqref{eq:comp.lie_tens} the orthogonal projection of the Lie derivative of $\bm k$ 
	along the evolution vector into 
	\beqNo
		\proj{\gamma}{\alpha} \proj{\delta}{\beta} \, \CL_{\bm m} k_{\gamma \delta} 
			= \shl \proj{\gamma}{\alpha} \proj{\delta}{\beta} 
				\, u^\lambda \, k_{\gamma \delta \, ; \, \lambda}
				- \shl k^\gamma_{\phantom{\gamma} \beta} \left( k_{\alpha \gamma} + k_{\gamma \alpha} \right) \, . 
	\eeqNo
	The left-hand side of this relation is reformulated by means of \eqref{eq:evol.3met.mixed} and the spatial character of $\bm k$ as 
	\beqNo
		 \proj{\gamma}{\alpha} \proj{\delta}{\beta} \, \CL_{\bm m} k_{\gamma \delta} 
		 	= \CL_{\bm m} k_\albe \, . 
	\eeqNo
	We use the resulting expression to substitute the first term of the right-hand side of \eqref{eq:dev.evol.ext_curv.I}, and we 
	conclude the proof.
\end{proof}
%

\subsection{\texorpdfstring{$1+3$}{1+3} form of the Ricci tensor} \label{subsec:proj.ricci}

\subsubsection{Contracted Gauss equation and scalar Gauss equation}

\begin{corollaryS} \label{cor:gauss_rel}
	In the 1+3 formalism the contracted Gauss relation is given by 
	\beq \label{eq:gauss_rel.II}
		\proj{\gamma}{\alpha} \proj{\delta}{\beta} {}^{4\!}R_{\gamma \delta} 
			+ \proj{\gamma}{\alpha} \proj{\delta}{\beta} u^\sigma u_\rho {}^{4\!}R^\rho_{\phantom{\rho} \gamma \sigma \delta} 
			= R_{\alpha \beta} + k k_{\alpha \beta} - k^\gamma_{\phantom{\gamma} \beta} k_{\alpha \gamma} \, , 
	\eeq
	and the scalar Gauss relation reads 
	\beq \label{eq:gauss_rel.III}
		{}^{4\!}R + 2 \, u^\alpha u^\beta \, {}^{4\!}R_\albe 
			= R + k^2 - k_\albe k^{\beta \alpha} \, . 
	\eeq
\end{corollaryS}
\begin{proof}
	The first equation is obtained by contracting \eqref{eq:gauss_rel.I} on the first and third indices and by using the idempotence 
	of the orthogonal projector together with \eqref{eq:comp.orth_proj}. The second equality stems from the trace of \eqref{eq:gauss_rel.II}.
\end{proof}
%

\subsubsection{Contracted Codazzi equation}

\begin{corollaryS} \label{cor:cod_rel}
	In the 1+3 formalism the contracted Codazzi relation is written 
	\beq \label{eq:cod_rel.II}
		\proj{\gamma}{\alpha} u^\delta \, {}^{4\!}R_{\gamma \delta} 
	 		= k_{\, | \, \alpha} - k^\gamma_{\phantom{\gamma} \alpha \, || \, \gamma} 
			- a^\gamma \left( k_{\gamma \alpha} - k_{\alpha \gamma} \right) \, .
	\eeq
\end{corollaryS}
\begin{proof}
	This expression is derived by contracting Eq.~\eqref{eq:cod_rel.I} on the indices $\mu$ and $\alpha$.
\end{proof}
\begin{remark}
	The contraction of \eqref{eq:cod_rel.I} on the indices $\mu$ and $\beta$ yields the same expression as \eqref{eq:cod_rel.II} 
	(cf.\ Remark~\ref{rem:cod}).
\end{remark}
%

\subsubsection{Reformulation of the Ricci equation}

\begin{corollaryS} \label{cor:ricci_rel}
	In the 1+3 formalism the Ricci equation can be alternatively written as 
	\beq \label{eq:ricci_rel.II}
		\proj{\mu}{\alpha} \proj{\nu}{\beta} {}^{4\!}R_\mnu 
			= R_\albe + k k_\albe 
			- k^\gamma_{\phantom{\gamma} \beta} \left( k_{\alpha \gamma} + k_{\gamma \alpha} \right) 
			- \frac{1}{\shl} \CL_{\bm m} k_\albe - a_{\alpha \, || \, \beta} - a_\alpha a_\beta \, . 
	\eeq
\end{corollaryS}
\begin{proof}
	This expression is obtained by combining Eq.~\eqref{eq:ricci_rel.I} with Eq.~\eqref{eq:gauss_rel.II}.
\end{proof}
\begin{remark}
	Taking the trace of \eqref{eq:ricci_rel.I} and using \eqref{eq:evol.3met.contrav} and \eqref{eq:gauss_rel.III} on the outcome, we get 
	\beqNo
		{}^{4\!}R 
			= R + k^2 + k_\albe k^{\beta \alpha} - \frac{2}{M} \CL_{\bm m} k
			- 2 a^\alpha_{\phantom{\alpha} || \, \alpha} - 2 \, a^\alpha a_\alpha \, . 
	\eeqNo
	(This relation is useful for instance for the derivation of the 1+3 form of the Hilbert action.)
\end{remark}
%

\section{\texorpdfstring{$1+3$}{1+3} form of Einstein's equation \label{sec:1+3.efe}}

We present in this section the 1+3 formulation of the Einstein field equation 
\beq \label{eq:efe}
	{}^{4\!}\bm R - \frac{1}{2} {}^{4\!}R \, \bm g + \Lambda \, \bm g
		= 8 \pi \bm T \, .
\eeq
$\bm T$ denotes the stress--energy tensor of the fluid filling the space--time, and $\Lambda$ is the cosmological constant that we carry 
along for the sake of generality. (We use geometric units, for which $G = c = 1$.)

We identify in what follows the four-dimensional Lo\-rent\-zian manifold to the physical space--time, the integral curves of the congruence 
to the world lines of the fluid and the local spatial frames to the (instantaneous) rest-frames of the fluid.

The 1+3 decomposition of the space--time Ricci tensor is given in the above section; to provide the mentioned reformulation, we are left 
with the derivation of the 1+3 form of the stress--energy tensor.

\subsection{\texorpdfstring{$1+3$}{1+3} form of the stress--energy tensor}

The stress--energy tensor of the fluid is decomposed in its rest-frames according to 
\beqNo
	\bm T 
		= \epsilon \, \bm u^\flat \otimes \bm u^\flat 
		+ \bm q \otimes \bm u^\flat + \bm u^\flat \otimes \bm q 
		+ p \, \bm \gamma + \bm \zeta \, . 
\eeqNo
It reads in component form 
\beqNo
	T_\mnu 
		= \epsilon \, u_\alpha u_\beta + 2 \, q_{(\alpha} u_{\beta)} + p \, \gamma_\albe + \pi_\albe \, , 
\eeqNo
where we have defined 
\beq \label{eq:def.dyn_fluid}
	\epsilon 
		:= u^\alpha u^\beta T_\albe \, , \qquad 
	q_\alpha 
		:= - \gamma^\mu_{\phantom{\mu} \alpha} u^\nu T_\mnu \, , \qquad
	p \, \gamma_\albe + \zeta_\albe 
		:= \gamma^\mu_{\phantom{\mu} \alpha} \gamma^\nu_{\phantom{\nu} \beta} T_\mnu \, .
\eeq
$\epsilon$ stands for the energy density of the fluid, $p$ for its isotropic pressure, $\bm q$ defines its (spatial) heat vector and $\bm \zeta$ 
its (spatial, symmetric and traceless) anisotropic stress tensor.

The trace of the stress--energy tensor is given by 
\beq \label{eq:trace.se_tens}
	T = - \epsilon + 3 p \, . 
\eeq
%

\subsection{\texorpdfstring{$1+3$}{1+3} form of Einstein's equation}

The trace of Einstein's equation yields the relation 
\beq \label{eq:trace.4ricci}
	{}^{4\!}R 
		= - 8 \pi T + 4 \Lambda \, , 
\eeq
which, inserted back into \eqref{eq:efe}, drives 
\beqNo
	{}^{4\!}\bm R 
		= 8 \pi \, \Big( \bm T - \frac{1}{2} \, T \bm g \Big) + \Lambda \bm g \, , 
\eeqNo
and in component form 
\beqNo
	{}^{4\!}R_\albe 
		= 8 \pi \, \Big( T_\albe - \frac{1}{2} \, T g_\albe \Big) + \Lambda g_\albe \, . 
\eeqNo
We project this last expression on the world lines and the rest-frames of the fluid. From Eqs.~\eqref{eq:def.dyn_fluid} and 
\eqref{eq:trace.se_tens} we get 
\begin{gather}
	u^\alpha u^\beta \, {}^{4\!}R_\albe 
		= 4 \pi \, (\epsilon + 3p) - \Lambda \, , 
		\label{eq:proj.ricci.uu} \\[1pt]
	\proj{\mu}{\alpha} u^\nu \, {}^{4\!}R_\mnu 
		= - 8 \pi q_\alpha \, , 
		\label{eq:proj.ricci.ug} \\[-1pt]
	\; \proj{\mu}{\alpha} \proj{\nu}{\beta} \, {}^{4\!}R_\mnu 
		= 4 \pi \, \Big( ( \epsilon - p ) \, \gamma_\albe + 2 \zeta_\albe \Big) + \Lambda \gamma_\albe 
		\label{eq:proj.ricci.gg} \, . 
\end{gather}
The left-hand sides are then reformulated by way of the above corollaries, and we end up with the following set of relations, equivalent 
to the Einstein field equation (see Section \ref{sec:conclusion} for a discussion).

\subsubsection{Einstein--Gauss equation}

\begin{propositionE} \label{prop:einstein-gauss}
	In the 1+3 formalism the Einstein--Gauss relation is given by 
	\beq \label{eq:einstein-gauss}
		R + k^2 - k_\albe k^{\beta \alpha} 
			= 16 \pi \epsilon + 2 \Lambda \, . 
	\eeq
\end{propositionE}
\begin{proof}
	This expression is derived from Eq.~\eqref{eq:proj.ricci.uu}, the scalar Gauss relation \eqref{eq:gauss_rel.III} and 
	Eq.~\eqref{eq:trace.4ricci}.
\end{proof}
\begin{remark}
	Eq.~\eqref{eq:einstein-gauss} is a scalar relation. It is equivalent to the two-temporal projection \eqref{eq:proj.ricci.uu}.
\end{remark}
%

\subsubsection{Einstein--Codazzi equation}

\begin{propositionE} \label{prop:einstein-codazzi}
	In the 1+3 formalism the Einstein-Codazzi relation is written 
	\beq \label{eq:einstein-codazzi}
		k^\gamma_{\phantom{\gamma} \alpha \, || \, \gamma} 
			- k_{ \, | \, \alpha} + a^\gamma \left( k_{\gamma \alpha} - k_{\alpha \gamma} \right)
			= 8 \pi q_\alpha \, . 
	\eeq	
\end{propositionE}
\begin{proof}
	This expression is obtained from Eq.~\eqref{eq:proj.ricci.ug} and the contracted Codazzi relation \eqref{eq:cod_rel.II}.
\end{proof}
\begin{remark}
	Eq.~\eqref{eq:einstein-codazzi} is an expression relating spatial 1-forms. It comprises three independent relations, equivalent to 
	those coming from the one-temporal, one-spatial projection \eqref{eq:proj.ricci.ug}.
\end{remark}
%

\subsubsection{Einstein--Ricci equation}

\begin{propositionE} \label{prop:einstein-ricci}
	In the 1+3 formalism the Einstein--Ricci relation is given by 
	\begin{align} \label{eq:einstein-ricci}
		\frac{1}{\shl} \CL_{\bm m} k_{(\albe)} =
			& - a_{( \alpha \, || \, \beta )} - a_{ ( \alpha} a_{ \beta ) } + R_{ ( \albe ) } + k k_{ ( \albe ) } 
				- k^\gamma_{\phantom{\gamma} (\beta} \, k_{\alpha) \gamma} 
				- k_{\gamma (\alpha} k^\gamma_{\phantom{\gamma} \beta)} \nonumber \\
			&  - 4 \pi \, \Big( ( \epsilon - p ) \, \gamma_\albe + 2 \, \zeta_\albe \Big) 
				- \Lambda \gamma_\albe \, . 
	\end{align}
\end{propositionE}
\begin{proof}
	The symmetric part of \eqref{eq:ricci_rel.II} is written 
	\begin{align} \label{eq:dev.einstein-ricci}
		\proj{\mu}{(\alpha} \proj{\nu}{\beta)} {}^{4\!}R_\mnu = 
			& \; R_{ ( \albe ) } + k k_{(\albe)} 
				- k^\gamma_{\phantom{\gamma} (\beta} \, k_{\alpha) \gamma} 
				- k_{\gamma (\alpha} k^\gamma_{\phantom{\gamma} \beta)} \nonumber \\
			& - \frac{1}{\shl} \CL_{\bm m} k_{ ( \albe ) } 
				- a_{ ( \alpha \, || \, \beta ) } - a_{ ( \alpha} a_{\beta ) } \, , 
	\end{align}
	where the parentheses imply symmetrization over the indices enclosed. Noticing that 
	\beqNo
		\proj{\mu}{( \alpha} \proj{\nu}{\beta )} \, {}^{4\!}R_\mnu
			= \proj{\mu}{\alpha} \proj{\nu}{\beta} \, {}^{4\!}R_\mnu \, , 
	\eeqNo
	because of the symmetry of the four-Ricci curvature, we combine \eqref{eq:dev.einstein-ricci} with \eqref{eq:proj.ricci.gg}, and 
	we conclude the proof.
\end{proof}
\begin{remark} \label{rem:einstein-codazzi}
	Eq.~\eqref{eq:einstein-ricci} is an expression written for \textit{symmetric} rank-2 spatial terms. It has six independent components, 
	equivalent to those given by the two-spatial projection \eqref{eq:proj.ricci.gg}.
\end{remark}
%

%
%
\vspace{0.2mm}
\begin{remark}
	With the the help of Eq.~\eqref{eq:prop.3ricci} we can write the symmetric part of the spatial Ricci curvature as 
	\beqNo
		R_{(\albe)} 
			= R_\albe + k k_{[\albe]} + k^\gamma_{\phantom{\gamma} [\alpha} k_{\beta] \gamma} \, . 
	\eeqNo
	By means of this expression and using the fact that both $k_{\gamma \alpha} k^\gamma {}_\beta$ and $a_\alpha a_\beta$ are 
	symmetric terms, we can reformulate \eqref{eq:einstein-ricci} into 
	\vspace{-1mm}
	\begin{align*}
		\frac{1}{\shl} \CL_{\bm m} k_{(\albe)} =
			& - a_{( \alpha \, || \, \beta )} - a_\alpha a_\beta + R_\albe + k k_\albe 
				- k^\gamma_{\phantom{\gamma} \beta} \, \big( k_{\alpha \gamma} + k_{\gamma \alpha} \big) \\
			&  - 4 \pi \, \Big( ( \epsilon - p ) \, \gamma_\albe + 2 \, \zeta_\albe \Big) 
				- \Lambda \gamma_\albe \, . 
	\end{align*}
\end{remark}
\begin{remark} \label{rem:einstein-codazzi.II}
	Combining directly \eqref{eq:proj.ricci.gg} with \eqref{eq:ricci_rel.II} results in the expression 
	\vspace{-1mm}
	\begin{align*}
		\frac{1}{\shl} \CL_{\bm m} k_\albe =
			& - a_{ \alpha \, || \, \beta } - a_\alpha a_\beta + R_\albe + k k_\albe 
				- k^\gamma_{\phantom{\gamma} \beta} \, \big( k_{\alpha \gamma} + k_{\gamma \alpha} \big) \\
			&  - 4 \pi \, \Big( ( \epsilon - p ) \, \gamma_\albe + 2 \, \zeta_\albe \Big) 
				- \Lambda \gamma_\albe \, . 
	\end{align*}
	This equation is made of nine independent components. Six of them are given by the symmetric part \eqref{eq:einstein-ricci} and 
	are equivalent to the projected Einstein expression \eqref{eq:proj.ricci.gg} (cf.\ Remark \ref{rem:einstein-codazzi}). The other three 
	are given by the anti-symmetric part 
	\vspace{-1mm}
	\beqNo
		\frac{1}{\shl} \CL_{\bm m} k_{ \, [ \albe ] } 
			= - a_{\, [ \alpha \, || \, \beta ] } \, , 
	\eeqNo
	which is a purely geometric expression. It can be indeed directly obtained by taking the anti-symmetric part of \eqref{eq:ricci_rel.II}. 
	Hence it does \textit{not} involve the Einstein field equation for its formulation.
\end{remark}
%

\section{Conclusion \label{sec:conclusion}}

\subsection{Summary}

We have presented in this paper the formalism for the splitting of a four-dimensional Lorentzian manifold by a set of time-like integral 
curves. From the introduction of the geometrical tensors characterizing the local spatial frames (namely, the two fundamental forms in 
Section \ref{sec:fund_forms}, the spatial covariant derivative in Section \ref{sec:3connect} and the spatial curvature tensors in Section 
\ref{sec:3curv}), we have derived in a general space--time basis and its dual the Gauss, Codazzi and Ricci equations, along with the evolution 
equation for the spatial metric (Section \ref{subsec:proj.riem}). These relations were given in Propositions \ref{prop:gauss_rel} to 
\ref{prop:ricci_rel} and were used subsequently to obtain the different projections of the four-Ricci tensor (Section \ref{subsec:proj.ricci}), 
which resulted in Corollaries \ref{cor:gauss_rel} to \ref{cor:ricci_rel}. From these last expressions, we have provided the 1+3 formulation 
of the Einstein field equation in Propositions \ref{prop:einstein-gauss} to \ref{prop:einstein-ricci} (Section \ref{sec:1+3.efe}), valid as well 
in any bases.

\subsection{Discussion}

\subsubsection{General remarks}

Let us classify for the sake of argument the differences between the 1+3 and 3+1 descriptions in two categories: their construction and 
their formalism.

The 1+3 procedure is built from a congruence of time-like integral curves, and it provides a global time-like relation between points. 
In this type of splitting, the \textit{direction of time} is chosen. The 3+1 procedure is based on the introduction of a family of space-like 
hypersurfaces, and it supplies a global space-like association of points. Here, the three \textit{directions of space} are chosen.

When both the orientations of time and space are selected, both splittings can apply. In the generic configuration, a set of time-like 
integral curves and an independent family of space-like surfaces cover the manifold. The simplest setting is made of one vorticity-free 
time-like congruence: the integral curves provide the orientation of time, and the space-like orthogonal frames, which globally form 
hypersurfaces, provide the orientation of space. On account of this particular configuration, the 3+1 description is sometimes regarded 
as being identical to the 1+3 description without vorticity, although they differ in their construction.

In terms of formalism, the key difference between the 1+3 and 3+1 descriptions comes from the properties of the extrinsic curvature tensor. 
It is symmetric in the latter description, while it contains an anti-symmetric part in the former. As it was established, this additional term 
brings about several (interrelated) effects that are absent in the 3+1 perspective: (i) it prevents the spatial frames to form hypersurfaces, 
(ii) it sources the temporal part of the Lie bracket of two spatial vectors, (iii) it induces a torsion for the spatial connection, and at last (iv) 
it calls for a redefinition of the spatial Riemann curvature.

As another difference, let us mention that in the 1+3 approach the acceleration of the flow vector cannot be written (solely) in terms of 
a gradient. In a general setting, no particular expression can be actually supplied. When choosing bases and coordinates adapted to the 
congruence, we find that the acceleration can be expressed in terms of a gradient \textit{plus} another term. This additional term, because 
of the non-zero anti-symmetric part of the extrinsic curvature, does not vanish (see Appendix \ref{app:bases&coord}).

\subsubsection{\texorpdfstring{$1+3$}{1+3} form of the curvature tensors}

The 1+3 formulation of the four-Riemann curvature is supplied in Propositions \ref{prop:gauss_rel}, \ref{prop:cod_rel} and 
\ref{prop:ricci_rel}. All other projections either vanish or come down to those given here, owing to the symmetries of the tensor. 
We can compare these expressions with their 3+1 analogues in two ways, either by regarding the extrinsic curvature as such or 
by considering its symmetric and anti-symmetric parts.\footnote{%
See, e.g., \cite{smarr78}, \cite{baumg10} and \cite{gourg:book} for the 3+1 expressions.
}

Adopting the first point of view, Propositions \ref{prop:gauss_rel} and \ref{prop:ricci_rel} are identical to their 3+1 counterparts, 
provided that the indices of the extrinsic curvature are well ordered. Proposition \ref{prop:cod_rel}, on the other hand, differs 
from its 3+1 counterpart by an additional term.

In the second point of view, the comparison is realized by decomposing the extrinsic curvature into its symmetric and anti-symmetric 
parts. The additional contributions are given, after expansion, by all the non-symmetric terms involving $\bm k$.


\subsubsection{Variation of the spatial metric}

The variation of the spatial metric along the congruence is supplied in Proposition \ref{prop:evol.3met}. The expression for the 
covariant components is given by Eq.~\eqref{eq:evol.3met.cov}, that of the contravariant components by Eq.~\eqref{eq:evol.3met.contrav}, 
and that of the mixed components by Eq.~\eqref{eq:evol.3met.mixed}.

Equation \eqref{eq:evol.3met.cov} is the only expression identical to its 3+1 counterpart in both points of view. This stems from 
the fact that the right-hand side only involves the symmetric part of the extrinsic curvature. Equations \eqref{eq:evol.3met.contrav} 
and \eqref{eq:evol.3met.mixed}, on the other hand, differ from their 3+1 counterparts as the acceleration of the flow vector cannot 
be expressed (solely) in terms of a gradient.

\subsubsection{\texorpdfstring{$1+3$}{1+3} form of the Einstein equation}

The 1+3 formulation of Einstein's equation is given by Propositions \ref{prop:einstein-gauss}, \ref{prop:einstein-codazzi} and 
\ref{prop:einstein-ricci}. The ten component relations are equivalent to the projected Einstein equations \eqref{eq:proj.ricci.uu}, 
\eqref{eq:proj.ricci.ug} and \eqref{eq:proj.ricci.gg}, and accordingly to the Einstein field equation \eqref{eq:efe}. 

Adopting the first point of view to conduct the comparison, Propositions \ref{prop:einstein-gauss} and \ref{prop:einstein-ricci} 
are equivalent to their 3+1 counterparts. Proposition \ref{prop:einstein-codazzi} on the other hand contains one additional term. 
In the second point view, Proposition \ref{prop:einstein-ricci} is also equivalent to its 3+1 counterpart, as all terms involved in 
the relation are symmetric.

The 1+3 system of Einstein equations \eqref{eq:einstein-gauss}, \eqref{eq:einstein-codazzi} and \eqref{eq:einstein-ricci} can be 
supplemented by an evolution equation for the anti-symmetric part of $\bm k$ (cf.\ Remark \ref{rem:einstein-codazzi.II}), 
\beq \label{eq:evol.antisym_k}
	\frac{1}{\shl} \CL_{\bm m} k_{ \, [ \albe ] } 
		= - a_{\, [ \alpha \, || \, \beta ] } \, , 
\eeq
and by a constraint equation also for the anti-symmetric part of $\bm k$, 
\beq \label{eq:const.antisym_k}
	k_{[ \albe \, || \, \mu ]} 
		= a_{ [ \mu} k_{\albe ] } \, .
\eeq
(This last relation is obtained upon taking the full anti-symmetric part of Eq.~\eqref{eq:cod_rel.I}.) The system can be in addition 
complemented by the once-contracted Bianchi identities, which provide two evolution and two constraint equations for the 
electric and magnetic parts of the Weyl tensor, and by the twice-contracted Bianchi identities, which yield from Einstein's equation 
the energy and momentum conservation laws for the fluid. The system of equations is then closed by giving an equation of 
state for the fluid.\footnote{%
The interested reader can find the formulation of these two additional sets of equations in, e.g., \cite{elst97,ellis98}.
}




\begin{acknowledgements}
	It is a pleasure to thank Donato Bini, Thomas Buchert, Chris Clarkson, Eric Gourgoulhon and Robert T.~Jantzen for discussion 
	and valuable comments. The author acknowledges support from the Claude Leon Foundation.
\end{acknowledgements}

\appendix 
\renewcommand{\theequation}{\thesection.\arabic{equation}} 
\makeatletter
	\@addtoreset{equation}{section}
\makeatother
\section{Kinematical formulation of the \texorpdfstring{$1+3$}{1+3} Einstein equations \label{app:kin_approach}}

In this appendix we reformulate the 1+3 Einstein equations, given by Propositions \ref{prop:einstein-gauss}, \ref{prop:einstein-codazzi} 
and \ref{prop:einstein-ricci}, in terms of the kinematical quantities of the fluid.

\subsection{Kinematical quantities of the fluid}

The covariant derivative of the fluid 1-form $\bm u^\flat$ can be decomposed into 
\beq \label{eq:decomp.fluid_form}
	\connect \bm u^\flat 
		= - \bm a^\flat \otimes \bm u^\flat + \frac{1}{3} \Theta \, \bm \gamma + \bm \sigma - \bm \omega \, ,
\eeq
or in component form 
\beqNo
	u_{\alpha \, ; \, \beta} 
		= - a_\alpha u_\beta + \frac{1}{3} \Theta \, \gamma_\albe + \sigma_\albe - \omega_\albe \, ,
\eeqNo
where we have defined 
\beq \label{eq:def.kin_fluid}
	\Theta 
		:= u^\alpha_{\phantom{\alpha} ; \, \alpha} \, , \qquad 
	\sigma_\albe 
		:= \proj{\gamma}{( \alpha} \proj{\delta}{\beta )} \, u_{\delta \, ; \, \gamma} 
		- \frac{1}{3} \Theta \, \gamma_\albe \, , \qquad
	\omega_\albe 
		:= \proj{\gamma}{[ \alpha} \proj{\delta}{\beta ]} \, u_{\delta \, ; \, \gamma} \, .
\eeq
$\Theta$ defines the expansion rate of the fluid, $\bm \sigma$ its (spatial, symmetric and traceless) shear tensor, and $\bm \omega$ 
its (spatial, anti-symmetric and traceless) vorticity tensor.\footnote{%
Henceforth the symbol $\bm \omega$ is employed for the (2-form) fluid vorticity, and no more for a generic 1-form as in the main text.
}

\subsection{Extrinsic curvature}

Inserting Eq.~\eqref{eq:decomp.fluid_form} into Eq.~\eqref{eq:prop.ext_curv.II}, we get 
\beq \label{eq:prop.ext_curv.III}
	\bm k
		= - \frac{1}{3} \Theta \, \bm \gamma - \bm \sigma + \bm \omega \, ,
\eeq
and in component form 
\beq \label{eq:comp.prop.ext_curv.III}
	k_\albe 
		= - \frac{1}{3} \Theta \, \gamma_\albe - \sigma_\albe + \omega_\albe \, .
\eeq
On using the symmetries of the fields $\bm \gamma$, $\bm \sigma$ and $\bm \omega$, we deduce from \eqref{eq:prop.ext_curv.III} 
\beqNo
	\bm k (\bm v, \bm w) - \bm k (\bm w, \bm v) 
		= 2 \, \bm \omega (\bm v, \bm w) \, , 
\eeqNo
for any spatial vectors $\bm v$ and $\bm w$. Plugging in turn this expression into \eqref{eq:prop.ext_curv.I} and \eqref{eq:exp.3torsion} 
respectively yields 
\beqNo
	\bm u \cdot [ \bm v, \bm w ] 
		= - 2 \, \bm \omega (\bm v, \bm w ) \, , \qquad
	\torsionf (\bm v, \bm w) 
		= - 2 \, \bm \omega (\bm v, \bm w ) \, \bm u \, . 
\eeqNo
In the kinematical formulation, the temporal part of the Lie bracket of two spatial vectors, or, equivalently, the torsion of the spatial 
connection, is induced by the fluid vorticity.

\subsection{\texorpdfstring{$1+3$}{1+3} kinematical form of the Einstein field equation}

With the help of \eqref{eq:comp.prop.ext_curv.III}, we reformulate the 1+3 Einstein equations \eqref{eq:einstein-gauss}, 
\eqref{eq:einstein-codazzi} and \eqref{eq:einstein-ricci} in terms of the kinematical quantities of the fluid. Introducing the rate of shear 
and the rate of vorticity of the fluid respectively as 
\beqNo
	\sigma^2 
		:= \frac{1}{2} \sqrt{\sigma_\albe \sigma^\albe} \, , \qquad\;
	\omega^2 
		:= \frac{1}{2} \sqrt{\omega_\albe \omega^\albe} \, , 
\eeqNo
and adopting the notation 
\beqNo
	{\mathscr T}_{ \langle \albe \rangle } 
		:= {\mathscr T}_{ ( \albe )} - \frac{1}{3} \, {\mathscr T} \gamma_\albe
\eeqNo
to denote the symmetric trace-free part of a rank-2 spatial tensor $\bm{\mathscr T}$, we write: 
\begin{itemize}
	\item[$\bullet$] the Einstein--Gauss relation \eqref{eq:einstein-gauss} as 
		\beqNo
			R + \frac{2}{3} \Theta^2 - 2 \sigma^2 + 2 \omega^2 
				= 16 \pi \epsilon + 2 \Lambda \, , 
		\eeqNo
	\item[$\bullet$] the Einstein--Codazzi relation \eqref{eq:einstein-codazzi} as 
		\beqNo
			\frac{2}{3} \Theta_{\, | \, \alpha} - \sigma^{\gamma}_{\phantom{\gamma} \alpha \, || \, \gamma} 
				+ \omega^{\gamma}_{\phantom{\gamma} \alpha \, || \, \gamma} + 2 a^\gamma \omega_{\gamma \alpha}
					= 8 \pi q_\alpha  \, ,
		\eeqNo
	\item[$\bullet$] the Einstein--Ricci relation \eqref{eq:einstein-ricci} as 
		\beqNo
			\frac{1}{\shl} \CL_{\bm m} \Theta 
				= - 4 \pi \left( \epsilon + 3 p \right) + \Lambda - \frac{1}{3} \Theta^2 - 2 \sigma^2 + 2 \omega^2 
				+ a^\alpha_{\phantom{\alpha} || \, \alpha} + a^\alpha a_\alpha \, , 
		\eeqNo
		for the trace-part, and 
		\beqNo
			\frac{1}{\shl} \CL_{\bm m} \sigma_\albe 
				= a_{ \langle \albe \rangle } + a_{\langle \alpha} a_{\beta \rangle} - R_{ \langle \albe \rangle } 
				- \frac{1}{3} \Theta \sigma_\albe 
				+  2 \sigma^\gamma_{\phantom{\gamma} \alpha} \sigma_{\beta \gamma} 
				+ 2 \sigma^\gamma_{\phantom{\gamma} ( \alpha} \omega_{\beta ) \gamma} 
				+ 8 \pi \zeta_\albe \, , 
		\eeqNo
		for the symmetric trace-free part.
\end{itemize}
These relations constitute the kinematical formulation of the system of 1+3 equations \eqref{eq:einstein-gauss}, \eqref{eq:einstein-codazzi} 
and \eqref{eq:einstein-ricci}. They are equivalent to (the ten components of) the projected Einstein expressions \eqref{eq:proj.ricci.uu}, 
\eqref{eq:proj.ricci.ug} and \eqref{eq:proj.ricci.gg}. (See the discussion in Section \ref{sec:conclusion} for further remarks.)

%
%
\smallskip
\begin{remarkapp}
	The kinematical formulation of the additional relations \eqref{eq:evol.antisym_k} and \eqref{eq:const.antisym_k} is given by 
	\beqNo
		\frac{1}{\shl} \CL_{\bm m} \omega_\albe 
			= - a_{\, [ \alpha \, || \, \beta ] } \, , \qquad 
		\omega_{[ \albe \, || \, \mu ]} 
			= a_{ [ \mu} \omega_{\albe ] } \, .
	\eeqNo
\end{remarkapp}
%

\section{Adapted bases and coordinates \label{app:bases&coord}}

In this appendix, we introduce a four-dimensional vector basis adapted to the congruence of curves along with the dual 1-form basis. 
The time-like components of spatial tensors and the decomposition of the four-metric are specified. We introduce subsequently a set 
of adapted coordinates and we briefly discuss the implications for the 1+3 formalism.

For additional information about the choice of bases in the threading point of view, as well as a presentation of adapted local coordinates, 
we refer the reader to, e.g., \cite{jantzen91,jantzen92,bini12,jantzen:book}.

\subsection{Bases adapted to the congruence}

\subsubsection{Construction}

The congruence of integral curves upon which the threading procedure is built is characterized by the tangent vector field $\bm u$. 
This provides the following natural choice for the four-dimensional basis $\{ \basis_\alpha \}$: 
\beq \label{eq:def.basis.cong_vect}
	\basis_0 := \shl \bm u \, , 
		\qquad
	\basis_i \cdot \basis_i > 0 \, . 
\eeq
$\shl$ is the threading lapse function, introduced earlier via the evolution vector $\bm m$. In the kinematical formulation, it relates 
the flow of an arbitrary parameter $t$ of a congruence line to the flow of the proper time $\tau$ of the fluid element moving along 
that line, 
\beq \label{eq:def.1+3_lapse}
	\frac{d\tau}{dt} =: \shl \, . 
\eeq
The basis vectors $\{ \basis_i \}$ are only required to be space-like (specifically, we do not ask them to be spatial). The form basis 
dual to \eqref{eq:def.basis.cong_vect} is given by 
\beq \label{eq:def.basis.cong_form}
	\basis^0 := - \shl^{-1} \bm u^\flat + \bm \shl \, , 
		\qquad\;
	\langle \basis^i, \bm u \rangle = 0 \, .  
\eeq
$\bm \shl$ defines the threading shift 1-form, which provides the freedom of choice for the spatial part of the time-like 1-form basis, 
\beqNo
	\orthproj (\basis^0) 
		= \bm \shl \, . 
\eeqNo
The spatial character of the forms $\{ \basis^i \}$ is implied by the dual condition $\langle \basis^i , \basis_0 \rangle %
= \kroneck{i}{0}$.

In these bases, the components of the flow vector $\bm u$ and those of its dual $\bm u^\flat$ are respectively written 
\beq \label{eq:comp.u}
	u^\alpha = \shl^{-1} ( 1, 0) \, , 
		\qquad 
	u_\alpha = \shl ( -1, \shl_i ) \, .
\eeq
\begin{remarkapp}
	It is worth mentioning the similarities with the construction of adapted bases in the slicing approach (see, e.g., \cite{gourg:book}). 
	Given the unit time-like vector $\bm n$ everywhere orthogonal to a family of space-like hypersurfaces, the bases are there 
	defined by 
	\beq \label{eq:def.basis.hyp_form}
		\nbasis^0 := - N^{-1} \bm n^\flat \, , 
			\qquad\; 
		\bm g^{-1} (\nbasis^i, \nbasis^i) > 0 \, , 
	\eeq
	for the 1-form basis, and 
	\beq \label{eq:def.basis.hyp_vect}
		\nbasis_0 := N \bm n + \bm N \, , 
			\qquad\; 
		\langle \nbasis^i, \bm n \rangle = 0 \, ,
	\eeq
	for the dual vector basis. $N$ is the slicing lapse function; it relates the flow of an arbitrary parameter $t$ of an integral curve of 
	tangent vector $\bm n$ to the flow of the proper time of the (Eulerian) observers moving along that line. $\bm N$ defines the 
	slicing shift vector, and it offers the freedom of choice for the hypersurface-tangent part of the time-like vector basis.
	In these bases, the components of $\bm n$ and $\bm n^\flat$ are respectively written 
	\beqNo
		n^\alpha	= N^{-1} ( 1, - N^i) \, , 
			\qquad 
		n_\alpha	= - N ( 1, 0 ) \, .
	\eeqNo
\end{remarkapp}
%

\subsubsection{Components of spatial tensors}

In the bases \eqref{eq:def.basis.cong_vect} and \eqref{eq:def.basis.cong_form}, the time-like covariant components of a spatial tensor 
$\bm T$ of type $(k,l)$ vanish, 
\beqNo
	T^{\, \alpha_1 \ldots \alpha_k}_{\phantom{\, \alpha_1 \ldots \alpha_k} \beta_1 \ldots \, 0 \, \ldots \, \beta_l} 
		= 0 \, , 
\eeqNo
and the time-like contravariant components are given by 
\beqNo
	T^{\, \alpha_1 \ldots \, 0 \, \ldots \alpha_k}_{\phantom{\, \alpha_1 \ldots \, 0 \, \ldots \alpha_k} \beta_1 \ldots \beta_l} 
	= \shl_i \;
		T^{\, \alpha_1 \ldots \, i \, \ldots \alpha_k}_{%
			\phantom{\alpha_1 \ldots \, i \, \ldots \alpha_k} \beta_1 \ldots \beta_l} \, . 
\eeqNo
%

\subsubsection{Decomposition of the four-metric}

The metric tensor $\bm g$ is decomposed in the basis \eqref{eq:def.basis.cong_form} according to 
\beqNo
	\bm g 
		= - \shl^2 \, \basis^0 \otimes \basis^0 
		+ 2 \shl^2 \shl_i \, \basis^0 \otimes \basis^i \\
		+ \left( \gamma_{ij} - \shl^2 \shl_i \shl_j \right) \basis^i \otimes \basis^j \, , 
\eeqNo
and its inverse $\bm{g}^{-1}$ is decomposed in the basis \eqref{eq:def.basis.cong_vect} as 
\beqNo
	\bm g^{-1} 
		= - \big( \shl^{-2} - \shl_i \shl^i \big) \, \basis_0 \otimes \basis_0 
		+ \, 2 \shl^i \, \basis_0 \otimes \basis_i 
		+ \gamma^{ij} \, \basis_i \otimes \basis_j \, . 
\eeqNo
\begin{remarkapp}
	Let us mention again the similarities with the slicing approach. In the bases \eqref{eq:def.basis.hyp_form} and 
	\eqref{eq:def.basis.hyp_vect}, the metric tensor is written 
	\beqNo
		\bm g
			= - \big( N^2 - N_i N^i \big) \, \nbasis^0 \otimes \nbasis^0 
			+ 2 N_i \, \nbasis^0 \otimes \nbasis^i 
			+ h_{ij} \, \nbasis^i \otimes \nbasis^j \, , 
	\eeqNo
	and its inverse is decomposed as 
	\beqNo
		\bm g^{-1} 
			= - N^2 \, \nbasis_0 \otimes \nbasis_0 
			+ 2 N^{-2} N^i \, \nbasis_0 \otimes \nbasis_i 
			+ \left( h^{ij} - N^{-2} N^i N^j \right) \nbasis_i \otimes \nbasis_j \, , 
	\eeqNo
	where $\bm h := \bm g + \bm n \otimes \bm n$ denotes the metric of the space-like hypersurfaces.
\end{remarkapp}
%

\subsection{Adapted coordinates}

\subsubsection{Construction}

We introduce the set of coordinates $(t,X^i)$ adapted to the congruence of curves as follows: 
\begin{itemize}
	\item[$\bullet$] the spatial coordinates $X^i$ label the integral curves of $\scC$, 
	\item[$\bullet$] the temporal coordinate $t$ labels a family of hypersurfaces $\{ \Sigma_t \}$.
\end{itemize}
The second item demands that the manifold is globally hyperbolic, which we suppose henceforth. As the congruence can exhibit vorticity, 
the space-like hypersurfaces cannot be orthogonal to $\bm u$. We shall denote by $\bm n$ their unit normal vector.

The vector basis associated with the coordinates $(t, X^i)$ is written 
\beq \label{eq:def.basis.coord_vect}
	\{ \bm \partial_t, \bm \partial_i \} \, .
\eeq
By definition, the time-like vector $\bm \partial_t$ is tangent to the line of constant spatial coordinates, hence it is tangent to the integral 
curves. We choose it to be 
\beq \label{eq:def.0basis.coord_vect}
	\bm \partial_t := M \bm u \, . 
\eeq
By definition, the space-like vectors $\{ \bm \partial_i \}$ are tangent to the line $t = \mathrm{const.}$, $\{ X^j = \mathrm{const.} \}_{i %
\neq j}$; therefore they are tangent to the hypersurfaces. We have 
\beqNo
	\langle \bm n^\flat, \bm \partial_i \rangle = 0 \, .
\eeqNo
The vectors $\{ \bm \partial_i \}$ are not spatial (in the present terminology) since we have $\bm u \neq \bm n$.

We write the 1-form basis dual to \eqref{eq:def.basis.cong_vect} as 
\beq \label{eq:def.basis.coord_form}
	\{ {\bf d}t, {\bf d}X^i \} \, . 
\eeq
Because the basis $\{ \bm \partial_t, \bm \partial_i \}$ is a particular case of \eqref{eq:def.basis.cong_vect}, we can identify its dual to 
\eqref{eq:def.basis.cong_form} and write accordingly 
\beq \label{eq:basis.coord_form}
	{\bf d}t := - M^{-1} \bm u^\flat + \bm M \, , \qquad 
	\langle {\bf d}X^i , \bm u \rangle = 0 \,  .
\eeq
This concludes the construction of coordinates adapted to the congruence and of their related bases.

\subsubsection{Lie derivative}

We now specify the form taken by the operator 
\beqNo
	\frac{1}{\shl} \CL_{\bm m} \, , 
	 	\;\; \mathrm{with} \;\;\,
	 	\bm m := M \bm u \, , 
\eeqNo
in the coordinate basis \eqref{eq:def.basis.coord_vect}. From Eq.~\eqref{eq:def.0basis.coord_vect}, the definition of the Lie derivative 
(see, e.g., \cite{gourg:book}) and the fact that $\bm \partial_t$ is a basis vector, we have 
\beq \label{eq:lie.temp_der}
	\frac{1}{\shl} \CL_{\bm m} 
		= \frac{1}{\shl} \partial_t \, . 
\eeq
Noticing that the partial derivative with respect to $t$ is evaluated along the curves of constant coordinates $X^i$, we can identify 
$\partial_t$ to the total (convective) derivative $d/dt$. Accordingly we can write, by means of \eqref{eq:def.1+3_lapse}, 
\beqNo
	\frac{1}{\shl} \CL_{\bm m} 
		= \frac{d}{d\tau} \, . 
\eeqNo
The choice of the coordinates $(t,X^i)$ and of their related bases is thus completely adapted to the congruence. In the 1+3 kinematical 
formulation, they supply the Lagrangian point of view for the description of the fluid dynamics.

\subsubsection{Vorticity and acceleration}

Let us finally turn to the component forms of the fluid vorticity and acceleration in the introduced bases. Special attention should 
be given to the former quantity, as it constitutes a central aspect in the 1+3 (kinematical) formalism (see the discussion in Section 
\ref{sec:conclusion}).

In the bases \eqref{eq:def.basis.coord_vect} and \eqref{eq:def.basis.coord_form} the components of the vorticity are written 
\beq \label{eq:comp.vort}
	\omega_{ij} 
		= \shl \left( \shl_{[ i} \shl_{j ], t} - \shl_{[i, j]} \right) \, .
\eeq
\begin{proof}
	From Eq.~\eqref{eq:def.kin_fluid} and Eq.~\eqref{eq:comp.cd_tens} written for $\bm u$, we have 
	\beqNo
		\omega_\albe 
			= \proj{\gamma}{[ \alpha} \proj{\delta}{\beta ]} 
				\left( \basis_\gamma (u_\delta) - \coef{\lambda}{\gamma}{\delta} u_\lambda \right) \, . 
	\eeqNo
	With the help of \eqref{eq:rel.coef} and using the fact that the structure coefficients of $\{ \bm \partial_t, \bm \partial_i \}$ 
	vanish, we obtain 
	\beqNo
		\omega_\albe 
			= \proj{\gamma}{[ \alpha} \proj{\delta}{\beta ]} \, \basis_\gamma (u_\delta) \, . 
	\eeqNo
	$\bm \omega$ being a spatial tensor, its time-like covariant components are zero in the basis \eqref{eq:def.basis.coord_form}. 
	Hence we consider 
	\beqNo
		\omega_{ij} 
			= \proj{\gamma}{[ i} \proj{\delta}{j ]} \, \basis_\gamma (u_\delta) \, . 
	\eeqNo
	From Eqs.~\eqref{eq:comp.orth_proj} and \eqref{eq:comp.u} we conclude the proof.
\end{proof}
\begin{remarkapp}
	Equation \eqref{eq:comp.vort} shows that a vanishing shift 1-form implies a vanishing vorticity (the converse is not true). 
	This can be understood in terms of the choice of the bases as follows. A zero shift implies that the flow form $\bm u^\flat$ 
	can be formulated as a gradient (cf.\ Eq.~\eqref{eq:basis.coord_form}). Hence its rotational vanishes and the vorticity cancels.
\end{remarkapp}
%
%
\smallskip

\noindent
At last, we write the components of the acceleration with respect to the bases \eqref{eq:def.basis.coord_vect} and 
\eqref{eq:def.basis.coord_form} as 
\beqNo
	a_i = \frac{D_i \shl}{\shl} + \shl_{i , t} \, . 
\eeqNo
\begin{proof}
	From Eqs.~\eqref{eq:comp.lie_tens} and \eqref{eq:comp.3connect} we have 
	\beq \label{eq:dev.comp.acc}
		\proj{\beta}{\alpha} \CL_{\bm m} u_\beta 
			= \shl a_\alpha - D_\alpha \shl \, . 
	\eeq
	The components $u_\alpha$ are written by means of \eqref{eq:basis.coord_form} as 
	\beq
		u_\alpha = \shl \left( - ({\bf d}t)_\alpha + \shl_\alpha \right) \, . 
	\eeq
	Using the definition of the Lie derivative (see, e.g., \cite{gourg:book}) and the fact that $\bm m = {\bm \partial_t}$ is a basis 
	vector and ${\bf d}t$ a basis 1-form, we find 
	\beqNo
		\proj{\beta}{\alpha} \CL_{\bm m} u_\beta 
			= \shl \proj{\beta}{\alpha} \CL_{\bm m} \shl_\beta \,  .
	\eeqNo
	From \eqref{eq:evol.3met.mixed} and the spatial character of $\bm \shl$ we reformulate the right-hand side and we insert the 
	outcome into \eqref{eq:dev.comp.acc} to obtain 
	\beqNo
		a_\alpha = \frac{D_\alpha \shl}{\shl} + \CL_{\bm m} \shl_\alpha \, .
	\eeqNo
	With the help of \eqref{eq:lie.temp_der} we conclude the proof.
\end{proof}
%


\end{document}